%% file: main.tex
\documentclass[11pt]{article}

\usepackage{times}

\usepackage{boxedminipage}
\usepackage{multirow}

\usepackage{amsmath}
\usepackage{amsthm}
\usepackage{wrapfig}
\usepackage{xspace}
\usepackage{euscript}
\usepackage{lineno}
\usepackage{url}
\usepackage{microtype}

\usepackage{soul}
\usepackage{graphicx}

\usepackage{bm} 
\usepackage{mathtools} 

\usepackage{color}
\usepackage[noend]{algorithmic}
\usepackage[plain]{algorithm}
\usepackage{lipsum}
\usepackage{enumitem}

\usepackage{array}

\newtheorem{lemma}{Lemma}

\usepackage[top=1.0in,bottom=1.0in,left=1in,right=1in]{geometry}
\usepackage{amsfonts}

\input{commands}

\renewcommand{\paragraph}[1]{\noindent\textbf{#1.}}

\newcommand{\footremember}[2]{%
   \footnote{#2}
    \newcounter{#1}
    \setcounter{#1}{\value{footnote}}%
}
\newcommand{\footrecall}[1]{%
    \footnotemark[\value{#1}]%
} 

\title{Maintaining Contour Trees of Dynamic Terrains}

\author{
Pankaj K. Agarwal\footnote{Work supported by NSF under grants CCF-09-40671, CCF-10-12254, and CCF-11-61359, by Grant 2012/229 from the U.S.-Israel Binational Science Foundation, and by an ERDC contract W9132V-11-C-0003.}${\,\,}^,$\footremember{sttr}{Supported by U.S. Army Research Office contract W911NF-13-P-0018.} \\ 
Computer Science \\
Duke University \\
\url{pankaj@cs.duke.edu}
\and
Lars Arge\footremember{madalgo}{Center for Massive Data Algorithmics, a Center of the Danish National Research Foundation.} \\
MADALGO\\
Aarhus University\\
\url{large@madalgo.au.dk}
\and
Thomas M\o{}lhave\footrecall{sttr}\\
SCALGO USA\\
\url{thomas@scalgo.com}
\and
Morten Revsb\ae{}k\footrecall{madalgo}\\
MADALGO\\
Aarhus University\\
\url{mrevs@madalgo.au.dk}
\and
Jungwoo Yang\footrecall{madalgo}\\
MADALGO\\
Aarhus University\\
\url{jungwoo@madalgo.au.dk}
}
\date{}
\begin{document}
\maketitle

\begin{abstract}
We consider maintaining the contour tree $\contourtree$ of a
piecewise-linear triangulation $\mesh$ that is the graph of a time
varying height function $\height: \mathbb{R}^2 \rightarrow
\mathbb{R}$. We carefully describe the combinatorial change in
$\contourtree$ that happen as $\height$ varies over time and how these
changes relate to topological changes in $\mesh$. We present a kinetic
data structure that maintains the contour tree of $h$ over time. Our
data structure maintains certificates that fail only when $h(v)=h(u)$
for two adjacent vertices $v$ and $u$ in $\mesh$, or when two saddle
vertices lie on the same contour of $\mesh$. A certificate failure is
handled in $\OhOf(\log(n))$ time. We also show how our data structure
can be extended to handle a set of general update operations on
$\mesh$ and how it can be applied to maintain topological persistence
pairs of time varying functions.
\end{abstract}

\newpage

\input{introduction}
\input{preliminaries}

\input{topology-morten}

\input{data-structure}

\section{Extensions and Applications}
In this section we describe how our data structure can be extended to
support a range of other triangulation update operations besides
$\changeheight(v,r)$. We also describe how the data structure can be
extended to maintain $\contourtree$ as the entire mesh varies over
time. Finally, we show how the data structure can be used to maintain
persistence pairs of $\height$ as it varies over time.

\input{update_mesh}

\subsection{Deforming the entire mesh}
Our data structure easily extends to maintaining $\contourtree$ as the
entire mesh varies over time since the topology changing events and
the way these are handled will be the same. Instead of maintaining
certificates involving only a single vertex $v$ of $\terrain$, we
maintain a certificate for every edge $(v,u)$ of $\terrain$ that fail
as $\height(v) = \height(u)$, and a certificate for every edge
$(\alpha,\beta)$ that fail as $\height(\alpha) = \height(\beta)$.

\subsection{Topological persistence}
Topological persistence was introduced by Edelsbrunner et al.~\cite{elz-tps-00}
and can roughly be defined as follows.  Suppose we sweep a horizontal
plane in the direction of increasing values of $\height$ and keep
track of connected components in $\sublevel$ while increasing
$\level$. A component of $\sublevel$ is started at a minimum vertex
and ends at a negative saddle vertex when it joins with an older
component. Similarly, a hole of $\sublevel$ is started at a positive
saddle and ends at a maximum vertex. Based on this it is possible to
define minimum-saddle and maximum-saddle \emph{persistence pairs}
between the critical vertex that starts a component or hole and the
one that ends it. The \emph{persistence value} of a persistence pair
is simply the height difference between the vertices i.e. it is the
difference between the height at which the corresponding component was
started and the height it was ended. In the following sections, we
will first describe the notion of join and split trees and how these
can be maintained during a $\changeheight(v,r)$ operation, and then we
will show how persistence pairs can be maintained using join and split
trees.

\subsubsection{Maintaining join and split tree}
As described the contour tree of $\terrain$ encodes the topological
changes in $\terrain_\level$ as we increase $\level$ from $-\infty$ to
$\infty$. The \emph{join tree} encodes a subset of these changes
i.e. the topological changes in $\terrain_{<\level}$ that occur as
components in $\terrain_{<\level}$ are created at minimum vertices of
$\terrain$ and merged at negative saddle vertices of
$\terrain$. Similarly, the \emph{split tree} encodes the subset of
topological changes in $\terrain_{>\level}$. As a result the events
that occur on the merge and split tree of $\terrain$ from a
$\changeheight$ operation are a subset of the events that occur on the
contour tree. For example, the events that occur on the join tree are
the events that involve the minima and negative saddle vertices of
$\terrain$ i.e. the negative interchange events, sign change events
and a subset of the local events. The algorithm described above can
therefore also be used to maintain the merge and split tree of
$\terrain$ the only difference being the reduced set of events. Note
that to actually detect sign change events we need to maintain
$\contourtree$ simultanously with $\jointree$ or augment the positive
saddle vertices of $\terrain$ to the edges of $\jointree$.

\subsubsection{Maintain topological persistence pairs}
Consider the join tree $\jointree$ of $\terrain$ and let $\alpha$ be a
negative saddle in $\jointree$. Assume that $\jointree'$ and
$\jointree''$ are subtrees of $\jointree$ rooted in the children of
$\alpha$ such that $x'$ and $x''$ are the lowest minimum nodes of
$\jointree'$ and $\jointree''$, respectively. Then the highest of $x'$
and $x''$ is paired with $\alpha$. Therefore, by augmenting every
negative saddle node $\alpha$ in $\jointree$ with a pointer to the
lowest minimum node in the subtree rooted in $\alpha$, we implicitly
represent the persistence pairs of $\terrain$. The extra information
augmented to $\jointree$ is straightforward to maintain during the
merge tree events described above. However, maintaining the extra
information in $\jointree$ creates the need for a new event that
occurs as the height of two minimum vertices becomes equal.

Let $x$ and $y$ be minimum vertices in $\terrain$ such that
$\heighttime{t_0}(x) = \heighttime{t_0}(y)$. Assume that
$\heighttime{t_0^-}(x) < \heighttime{t_0^-}(y)$, then nothing happens
unless the least common ancestor $\textsc{LCA}(x,y)$ of $x$ and $y$ in
$\jointree$ stores a pointer to $y$. If this is the case we must
update $\textsc{LCA}(x,y)$ and any direct ancestor that also points to
$x$. If the merge tree is represented as a dynamic tree, then this is
straightforward to do in $\OhOf(\log(n))$ time.


\input{main.bbl}
\appendix

\end{document}

%% file: commands.tex
\def\terrain{\Sigma}

\def\mesh{\mathbb{M}}
\def\height{h}
\def\star{\mathrm{St}}  
\def\link{\mathrm{Lk}}

\def\descendtree{\Pi^{\downarrow}}
\def\ascendtree{\Pi^{\uparrow}}
\def\lowerlink{\mathrm{Lk^{-}}} 
\def\upperlink{\mathrm{Lk^{+}}} 
\def\level{\ell}
\def\levelset{\mesh_{\ell}}
\def\contourtree{\mathbb{T}}
\def\jointree{\mathbb{J}}

\def\retract{\rho}
\def\jointree{\EuScript{J}}

\def\sublevel{\mesh_{< \level}} 
\def\superlevel{\mesh_{> \level}}

\def\changeheight{\textsc{ChangeHeight}}

\def\insertv{\textsc{Insert}}
\def\delete{\textsc{Delete}}
\def\edgeflip{\textsc{EdgeFlip}}
\def\plane{\mathbb{R}^2}
\def\seqedge{\mathbb{E}}

\def\vertices{V}
\def\edges{E}
\def\faces{F}

\newcommand{\OhOf}{\mathrm{O}}

\newcommand{\thomas}[1]  {}
\newcommand{\pankaj}[1]  {}
\newcommand{\morten}[1]  {}
\newcommand{\lars}[1] {}
\newcommand{\jungwoo}[1]{ }

\input{ipesharecommands}

%% file: ipesharecommands.tex
\newcommand{\heighttime}[1]{\height_{#1}}
\newcommand{\contourdeform}[0]{H}
\newcommand{\cycle}[0]{K} 
\newcommand{\contour}[0]{C} 

\newcommand{\ada}[0]{\xi}
\newcommand{\adb}[0]{\eta}
\newcommand{\bua}[0]{\zeta}
\newcommand{\bub}[0]{\omega}

\newcommand{\R}{\mathbb{R}}


%% file: introduction.tex
\section{Introduction}
Within computational geometry, GIS and spatial databases there has
been extensive work on developing terrain algorithms for modeling and
analyzing the surface of the earth represented as a piecewise-linear
triangulated mesh $\mesh$ also known as a \emph{triangulated irregular
  network} (TIN). This mesh can be regarded as the graph of a
piecewise-linear \emph{height function} $\height: \mathbb{R}^2
\rightarrow \mathbb{R}$. Some of the main application areas have been
flood risk analysis, visibility analysis and visualization. In the
last decade this work has been increasingly fuelled by significant
advancements within remote sensing technologies such as LiDAR. These
technologies enable mapping the surface of the earth in increasingly
high (submeter) resolution for very large areas which has enabled
increasingly detailed analysis and visualisation.

The surface of the earth is continuously changing as the result of
both natural processes and human activity. Increasingly, this
continuous change is captured by surface representations due to
technological development within remote sensing and unmanned aerial
vehicles that enable low cost, rapid and repeated mapping of the
surface of the earth. This development transforms $\mesh$ from a
static to a time-varying object. As $\mesh$ varies, any derived
analysis result has to be updated accordingly to ensure the veracity
of the analysis. The frequency of surface updates and the size of
surface representations prohibits simply recomputing information
derived from the terrain as updates appear. This opens an exciting new
frontier within the study of algorithms for terrain analysis where the
objective is to efficiently update analysis results as $\mesh$ varies
over time.

Many terrain analysis algorithms rely on explicitly representing and
analysing the topology of $\mesh$~\cite{danner:terrastream,
  arge:socg2010,carr:cg2010}. The most general and widely used
representation of terrain topology is the \emph{contour tree}. Let
$\levelset$ be the intersection of $\mesh$ with a horizontal plane at
height $\level$, then intuitively the contour tree $\contourtree$
represents how the topology of $\levelset$ evolves as $\level$ goes
towards infinity i.e. how the connected components (\emph{contours})
of $\levelset$ appear, merge, split and disappear (see
Section~\ref{sec:preliminaries} for a formal definiton). In this paper
we will show how to efficiently update $\contourtree$ as $\mesh$
varies over time. Since the contour tree is a widely used building
block in many modeling and analysis algorithms this provides an
important step towards efficiently maintaining the output of these
algorithms as $\mesh$ varies over time.


\subsection{Related Work}
The first efficient algorithm for constructing contour trees of
piecewise-linear height functions on $\mathbb{R}^2$ was given by Van
Kreveld et al.~\cite{kobps-ctsss-97} and used $\OhOf(n\log(n))$
time. This algorithm was later extended to $\mathbb{R}^3$ by Tarasov
and Vyalyi in~\cite{tv-cct3s-98}, and to arbitrary dimensions by Carr
et al.\cite{Carr:03}. An out-of-core or socalled I/O-efficient
algorithm for constructing contour trees of terrain representations
that does not fit in main memory, was given by Agarwal et
al.~\cite{Agarwal:06}. The \emph{Reeb Graph} is a generalization of
the contour tree to manifolds of any dimension. Algorithms have also
been presented for efficiently constructing Reeb
Graphs.

The study of kinetic data structures that track attributes of time
varying geometric systems is well established and a substantial amount
of work has been performed within this field. Refer to
\cite{guibas:cghandbook} for a survey of this work. Specifically there
has been recent work on maintaining Reeb Graphs and Contour Trees of
time varying manifolds. In \cite{edelsbrunner:socg2008}, Edelsbrunner
et al. describe an algorithm for maintaining the Reeb Graph of time
varying 3-manifolds. They show that if $h$ is a smooth function, then
the combinatorial structure of the Reeb Graph only changes in discrete
events when (i) a critical point $u$ becomes degenerate (i.e. the
Hessian at $u$ becomes singular), or (ii) $\height(u)=\height(v)$ for
two saddle points $u$ and $v$ and both $u$ and $v$ lie on the same
contour. If $\height$ is a piece-wise linear function, (i) corresponds
to two adjacent vertices $u$ and $v$ of $\mesh$ with $\height(u) =
\height(v)$, and (ii) corresponds to two saddle vertices lying on the
same contour. Edelsbrunner et al. maintain a set of certificates that
fail only when the combinatorial structure of the Reeb Graph
changes. For every certificate failure their algorithm requires
$\OhOf(n)$ time to restore the combinatorial structure, where $n$ is
the number of vertices in $\mesh$. Wang and
Safa~\cite{safa-mpcttvfm-14} suggest an algorithm for maintaining
contour trees of time varying piecewise linear 2-manifolds. This
algorithm handles certificate failures in $\OhOf(\log(n))$ time,
however, they need to process a much larger number of certificate
failures since a certificate fails each time any two vertices of
$\mesh$ lie on the same contour. Their algorithm also works for simple
3-manifolds where the Reeb Graph is a contour tree.

\subsection{Our Results}
We describe an algorithm for maintaining contour trees of time varying
2-manifolds. Our algorithm processes certificate failures in
$\OhOf(\log(n))$ time and certificates only fail when $\height(u) =
\height(v)$ for two adjacent vertices $u$ and $v$ in $\mesh$, or when
two saddle vertices lie on the same contour. We maintain an auxiliary
data structure that maps a vertex in $\mesh$ to its corresponding edge
in $\contourtree$ in $\OhOf(\log(n))$ time. This data structure is
maintained through certificates that fail as $\height(u) = \height(v)$
for two adjacent vertices $u$ and $v$ in $\mesh$.

For simplicity, we focus the description of our data structure on the
operation $\changeheight(v,r)$ that change the height of $v$ in
$\mesh$ to $r$. Our approach easily generalizes to a setting where all
of $\height$ varies. We provide a very detailed description of the
combinatorial changes that occur in $\contourtree$ as the height of
$v$ varies and how these changes relate to topological changes in
$\height$. Specifically, we describe how the color of a contour on
$\height$ transitions during combinatorial changes in
$\contourtree$. We show how our data structure can be used to maintain
$\contourtree$ under an extended set of operations on $\mesh$ such as
vertex insertion, vertex deletion and edge flip. Finally, we show that
our algorithm can be used to maintain topological persistence pairs of
$\height$ as $\height$ varies over time.

%% file: preliminaries.tex
\section{Preliminaries} \label{sec:preliminaries}

\paragraph{Terrains}
Let $\mesh = (\vertices, \edges, \faces)$ be a triangulation of $\mathbb{R}^2$,
with vertex, edge, and face (triangle) sets $V$, $E$, and $F$, respectively,
and let $n=|V|$. We assume that $\vertices$ contains a vertex $v_{\infty}$ at
infinity, and that each edge $\{u, v_\infty\}$ is a ray emanating from $u$; the
triangles in $\mesh$ incident to $v_{\infty}$ are unbounded. Let $\height:
\mathbb{R}^2 \rightarrow \mathbb{R}$ be a continuous {\em height function} with
the property that the restriction of $\height$ to each triangle of $\mesh$ is a
linear map that for unbounded triangles approaches $-\infty$ at $v_{\infty}$
such that $h(v_{\infty}) = -\infty$. Given $\mesh$ and $\height$, the graph of
$\height$, called a {\em terrain} and denoted by $\terrain$, is as a
$xy$-monotone triangulated surface whose triangulation is induced by $\mesh$.
That is, vertices, edges, and faces of $\terrain$ are in one-to-one
correspondence with those of $\mesh$ and with a slight abuse of terminology we
refer to $\vertices$, $\edges$, and $\faces$, as vertices, edges, and triangles
of both the terrain $\terrain$ and $\mesh$. We assume that $\height(u) \neq
\height(v)$ for vertices $u,v \in V$ such that $u \neq v$.

\paragraph{Critical points}
For a vertex $v$ of $\mesh$, the \emph{star} of $v$, denoted by $\star(v)$,
consists of all triangles incident to $v$. The \emph{link} of $v$, denoted by
$\link(v)$, is the boundary of $\star(v)$, i.e. the cycle formed by edges that
are not incident on $v$ but belong to triangles that are in $\star(v)$. The
lower (resp. upper) link of $v$, $\lowerlink(v)$ (resp. $\upperlink(v)$), is
the subgraph of $\link(v)$ induced by vertices $u$ with $\height(u) <
\height(v)$ (resp. $\height(u) > \height(v)$). 

A \emph{minimum} (resp. \emph{maximum}) of $\mesh$ is a vertex $v$ for
which $\lowerlink(v)$ (resp. $\upperlink(v)$) is empty. A maximum or a
minimum vertex is called an \emph{extremal} vertex. A non-extremal
vertex $v$ is \emph{regular} if $\lowerlink(v)$ (and also
$\upperlink(v)$) is connected, and saddle otherwise; see
Figure~\ref{fig:vertex-types}. A vertex that is not regular is called
a \emph{critical} vertex. For simplicity, we assume that each saddle
vertex $v$ is \emph{simple}, meaning that $\lowerlink(v)$ and
$\upperlink(v)$ consists of only two components.

\paragraph{Level sets and contours}
For $\level \in \mathbb{R}$, the {\em $\level$-level set\/}, {\em
$\level$-sublevel set\/} and {\em $\level$-superlevel set\/} of $\mesh$ are
subsets $\levelset$, $\sublevel$, $\superlevel$ of $\plane$ consisting of
points $x$, with $\height(x) = \level$, $\height(x) < \level$, and $\height(x) >
\level$ respectively. Similarly, the closed {\em $\level$-sublevel\/} (resp.
{\em $\level$-superlevel\/}) set of $\mesh$ consists of points in $\plane$
with $\height(x) \leq \level$ (resp. $\height(x) \geq \level$). We will refer
to a level set $\levelset$ where $\level = \height(v)$ for some critical vertex
$v$ as a \emph{critical level}. A {\em contour} of $\mesh$ is a connected
component of a level set of $\mesh$.  Each vertex $v \in V$ is contained in
exactly one contour in $\mesh_{h(v)}$, which we call {\em the contour of
$v$}. A contour not passing through a critical vertex is a simple polygonal
cycle with non-empty interior. A contour passing through an extremal vertex is
a single point, and by our assumption, a contour passing through a saddle
consists of two simple cycles with the saddle vertex being their only
intersection point. A contour $C$ not passing through a vertex can be
represented by the circular sequence of edges of $\mesh$ denoted by
$\seqedge(C)$ that it passes through. Two contours are called
\emph{combinatorially identical} if their cyclic sequences are the same.

Let $\epsilon = \epsilon(\terrain)$ denote a sufficiently small
positive value, in particular, smaller than height difference between
any two vertices of $\terrain$.  An \emph{up-contour} of a saddle
vertex $\alpha$ is any contour of
$\mesh_{\height(\alpha)+\epsilon}$ that intersects an edge incident
on $\alpha$. Similarly, a \emph{down-contour} of $\alpha$ is any
contour of $\mesh_{\height(\alpha)-\epsilon}$ that intersects an
edge incident on $\alpha$. If $\alpha$ has two up-contours and one
down-contour it is called a \emph{positive saddle vertex}. If it has
two down-contours and one up-contour it is called a \emph{negative
  saddle vertex}. All simple saddles are either negative or positive.

\paragraph{Red and blue contours and saddle vertices}
A contour $\contour$ of $\levelset$ is called \emph{blue} if points
locally in the interior of $\contour$ belong to $\sublevel$ and \emph{red}
otherwise. We associate a color with a positive (resp. negative)
saddle vertex given by the color of its unique down-contour
(resp. up-contour). Refer to Figure \ref{figure:saddle-types} to see
the possible saddle colors.
\begin{figure}
\centering
\begin{tabular}{cccc}
\includegraphics[width=0.1\textwidth,page=1]{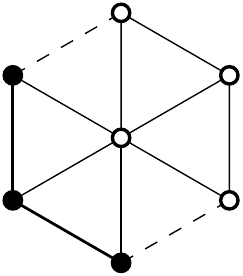} 
& \includegraphics[width=0.1\textwidth,page=2]{figures/vertex-types.pdf} 
& \includegraphics[width=0.1\textwidth,page=3]{figures/vertex-types.pdf} 
& \includegraphics[width=0.1\textwidth,page=4]{figures/vertex-types.pdf} \\
regular & minimum & saddle & maximum  
\end{tabular}
\caption{The types of a vertex; lower (resp. upper) link vertices are depicted by filled (resp. hollow) circles.}
\label{fig:vertex-types}
\end{figure}

\begin{figure}
\centering
\begin{tabular}{|>{\centering\arraybackslash}m{2.3cm}|>{\centering\arraybackslash}m{2.3cm}|>{\centering\arraybackslash}m{4.8cm}|>{\centering\arraybackslash}m{4.8cm}|}
\hline
Type & Contour tree & Terrain & Contour \\
\hline
Red Positive
& \includegraphics[width=0.1\textwidth,page=2]{figures/saddle-types.pdf}
& \includegraphics[width=0.3\textwidth]{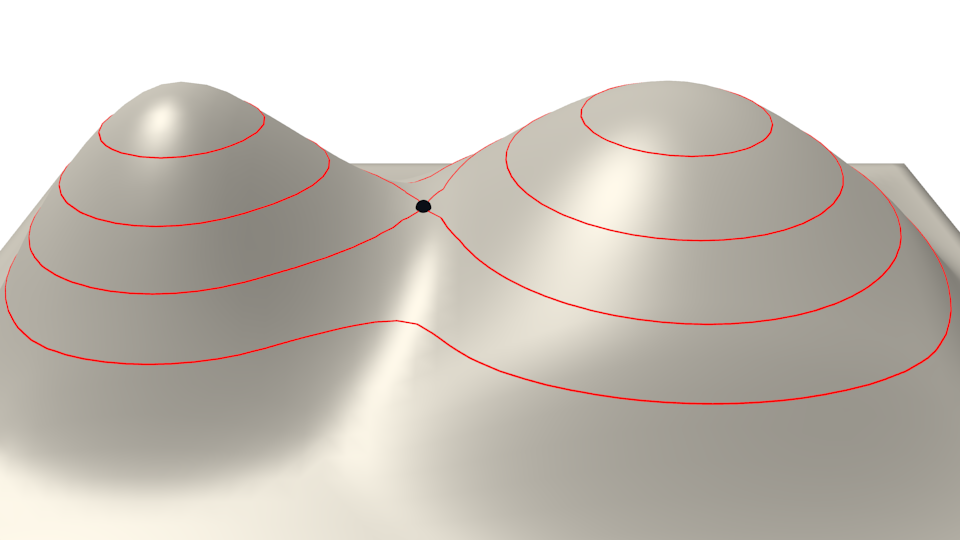}
& \includegraphics[width=0.3\textwidth]{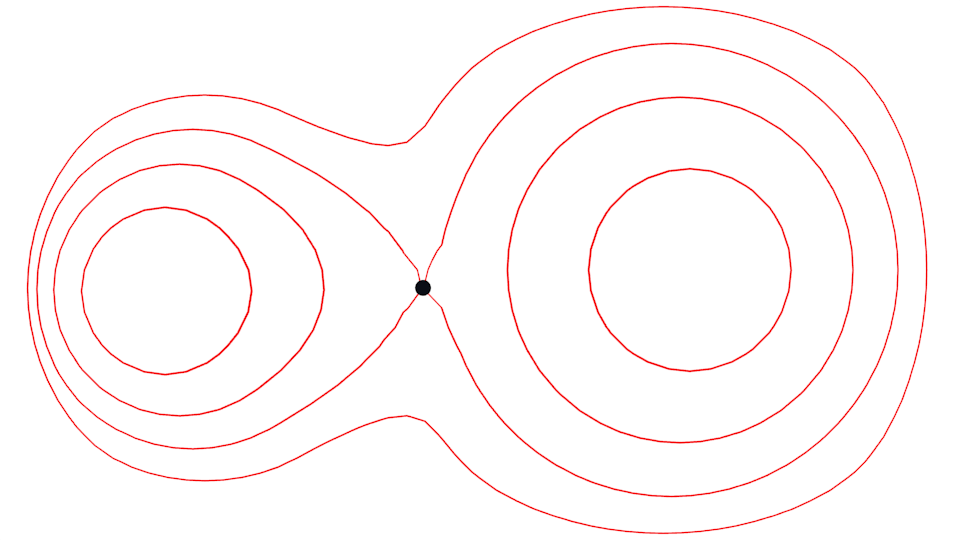} \\
\hline
Blue Positive
& \includegraphics[width=0.1\textwidth,page=3]{figures/saddle-types.pdf}
& \includegraphics[width=0.3\textwidth]{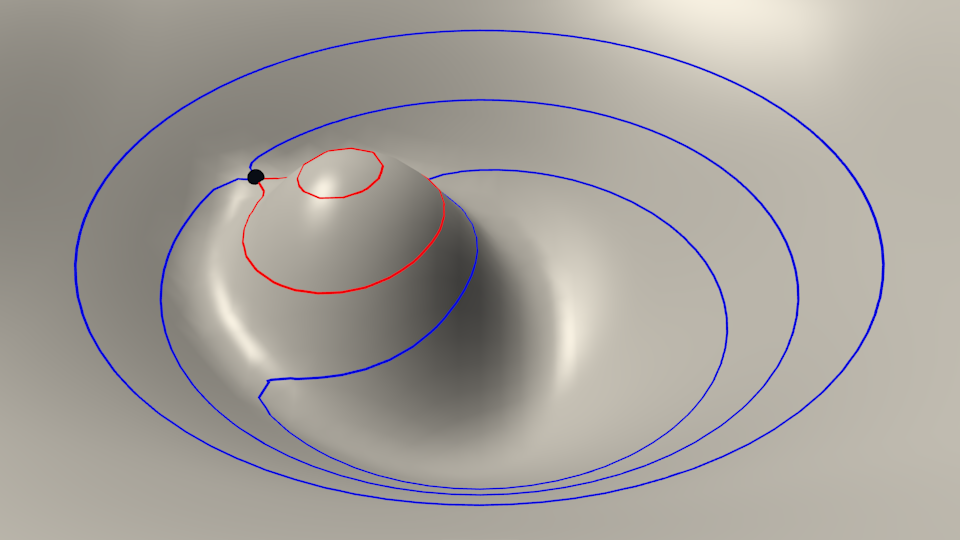}
& \includegraphics[width=0.3\textwidth]{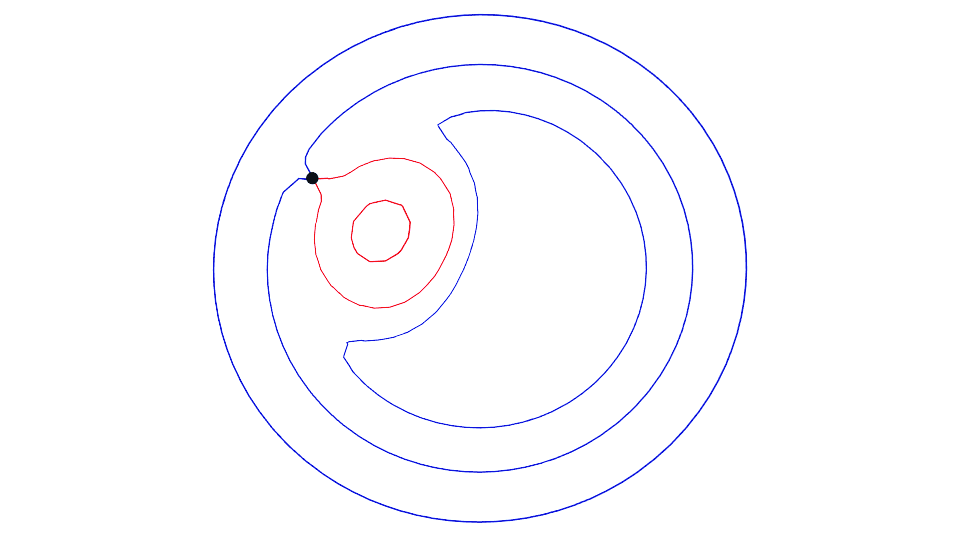} \\
\hline
Blue Negative
& \includegraphics[width=0.1\textwidth,page=4]{figures/saddle-types.pdf}
& \includegraphics[width=0.3\textwidth]{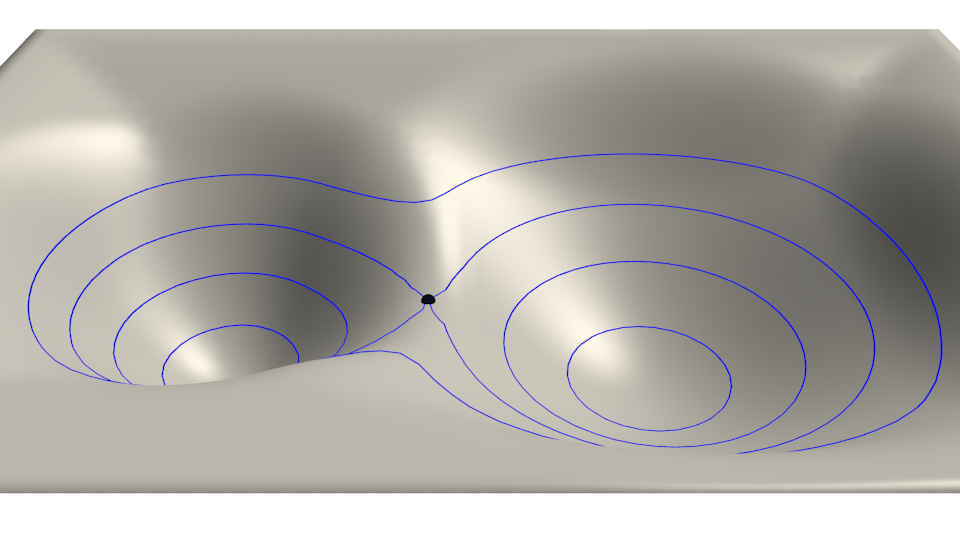}
& \includegraphics[width=0.3\textwidth]{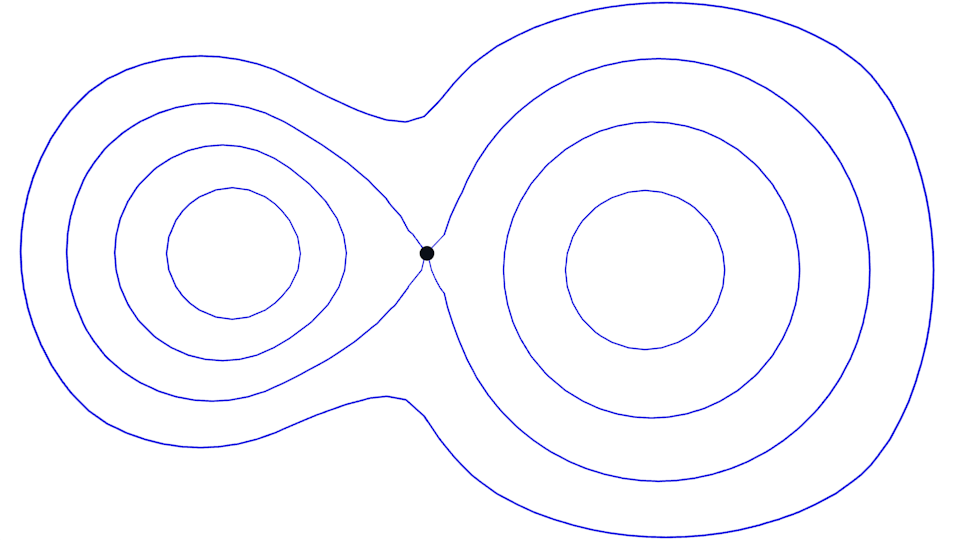} \\
\hline
Red Negative
& \includegraphics[width=0.1\textwidth,page=5]{figures/saddle-types.pdf}
& \includegraphics[width=0.3\textwidth]{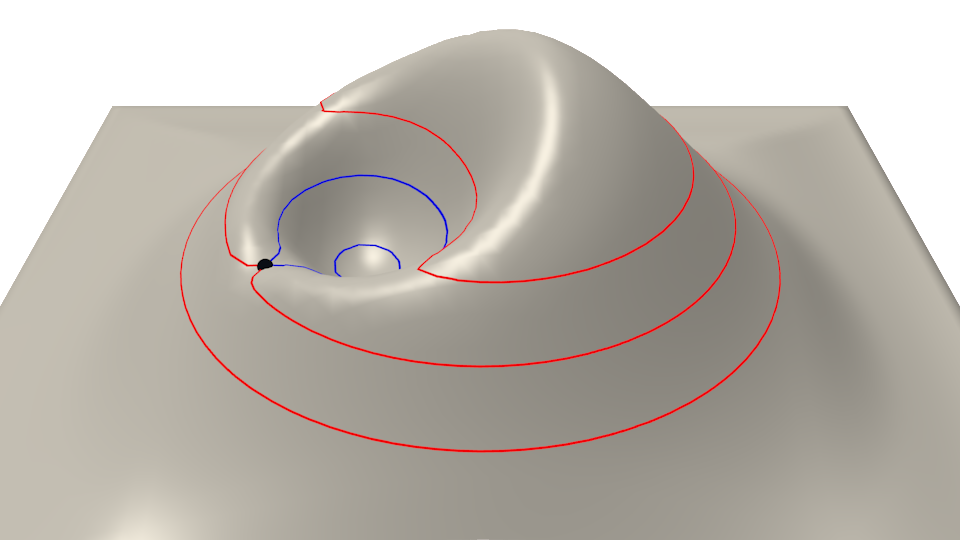}
& \includegraphics[width=0.3\textwidth]{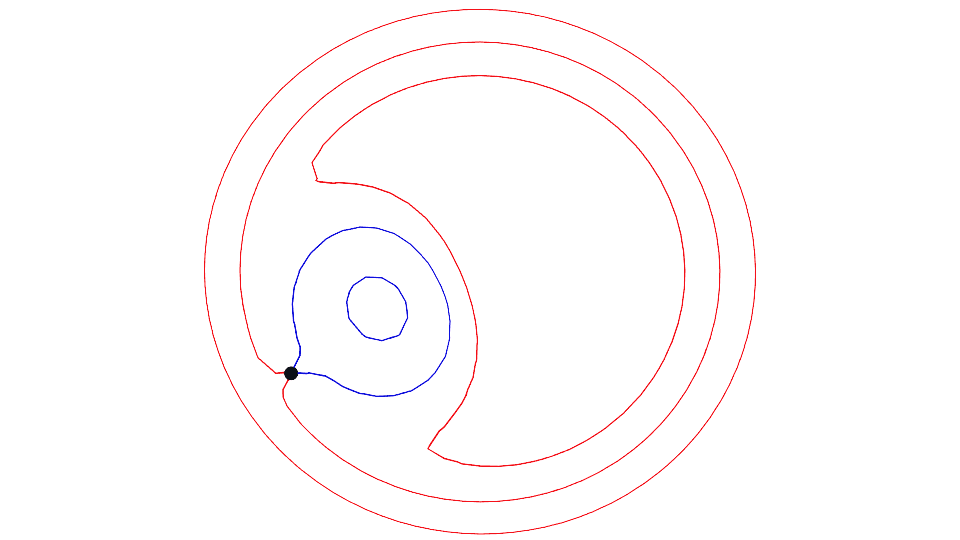} \\
\hline
\end{tabular}
\caption{Saddle Types}
\label{figure:saddle-types}
\end{figure}

\paragraph{Contour trees}
Consider raising $\level$ from $-\infty$ to $\infty$, between critical
levels the contours continuously deform but no changes happen to the
topology of the level set. At a minimum vertex, a new contour is
created. At a maximum vertex an existing contour contracts into a
single point and disappears. At a positive saddle vertex $v$, an
existing contour (the down-contour of $v$) is split into two new
contours (the up-contours of $v$) and at a negative saddle vertex $v$,
two contours (the down-contours of $v$) merge into one contour (the
up-contour of $v$).  The \emph{contour tree} $\contourtree$ of
$\terrain$ is a tree on the critical vertices of $\terrain$ that
encodes these topological changes of the level set. An edge $(v,w)$ of
$\contourtree$ \emph{represents} the contour that appears at $v$ and
disappears at $w$. 

More formally, two contours $\contour_1$ and $\contour_2$ at levels
$\ell_1$ and $\ell_2$, respectively, are called $\emph{equivalent}$ if
$\contour_1$ and $\contour_2$ belong to the same connected component
of $\Gamma = \{x \in \mathbb{R}^2 \mid \ell_1 \leq \height(x) \leq
\ell_2\}$ and $\Gamma$ does not contain any critical
vertex. An equivalence class of contour starts and ends at critical
vertices. If it starts at a critical vertex $v$ and ends at $w$, then
$(v,w)$ is an edge in $\contourtree$. We refer to $v$ as a \emph{down
  neighbor} of $w$ and to $w$ as an \emph{up neighbor} of
$v$. Equivalently $\contourtree$ is the quotient space in which each
contour is represented by a point and connectivity is defined in terms
of the quotient topology. Let $\retract : \mesh \rightarrow
\contourtree$ be the associated quotient map, which maps all points of
a contour to a single point on an edge of $\contourtree$. Note that
for any point $p$ on $\mesh$ that does not correspond to a critical
vertex, $\retract(p)$ is interior (not an endpoint) of a single edge
in $\contourtree$, if $p$ corresponds to an extremum vertex
$\retract(p)$ is the endpoint of a single edge in $\contourtree$, and
if $p$ corresponds to a saddle vertex then $\retract(p)$ is the
endpoint of several edges of $\contourtree$.

We will assume that each vertex of the contour tree
is labeled with the corresponding critical vertex of $\terrain$. The
combinatorial description of the contour tree is the set of vertices,
their labels, and the set of edges. It will also be convenient to
think of the contour tree embedded in $\terrain$, where the
coordinates of a vertex of the tree is the same as those of the
corresponding vertex in $\terrain$.  With a slight abuse of notation
we do not distinguish between the combinatorial structure and its
embedding in 3D.

\begin{figure}
\centering{
\begin{tabular}{ccc}
\includegraphics[width=0.4\textwidth]{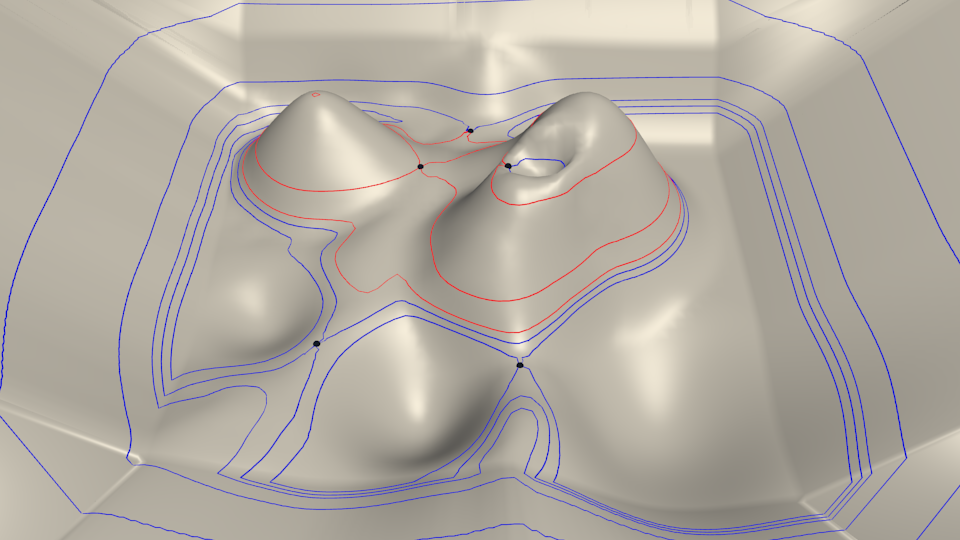} &
\includegraphics[width=0.35\textwidth,page=3]{figures/natural-terrain.pdf} &
\includegraphics[width=0.18\textwidth,page=5]{figures/natural-terrain.pdf} \\
(a) & (b) & (c)
\end{tabular}
}
\caption{(a) (b) An example terrain depicted with contours through saddle
  vertices and showing the critical vertices of the terrain. (c) The
  contour tree of the terrain in (a).}
\label{fig:terrain_min_max}
\end{figure}

\paragraph{Ascent and descent trees}
Consider the graph on the vertices of $\terrain$ given by creating an
edge $(w,v)$ from each vertex $v$ to a single vertex $w$ in
$\lowerlink(v)$ unless $v$ is a minimum vertex, in which case no edge
is created. This graph is then a forest of trees rooted in the minimum
vertices of $\terrain$. We refer to these trees as \emph{descent
  trees} and denote the descent tree rooted in $x$ as
$\descendtree(x)$. Similarly, we can define a forest of \emph{ascent
  trees} given by creating the edge $(v,w)$ from every vertex $v$ to a
single vertex $w$ in $\upperlink(v)$ unless $v$ is a maximum, in which
case no connection is made. Each ascend tree is rooted in a maximum
vertex $y$ and denoted $\ascendtree(y)$. Note that both the descent
tree forest and the ascent tree forest partitions the vertices in
$\terrain$ and that every vertex belongs to exactly one descent tree
and one ascent tree.

%% file: topology-morten.tex
\section{Continuous Height Change} \label{sec:deformation}
Our algorithm maintains the contour tree $\contourtree$ of $\terrain$
while allowing the height of vertices in $\terrain$ to be
changed. More formally, we support the $\changeheight(v,r)$ operation,
that given a vertex $v$ in $\terrain$ changes the height of $v$ to $r$
while updating $\contourtree$ to reflect any topological changes that
might be caused by the height change.  Besides $\contourtree$, we also
maintain a descent tree $\descendtree(x)$ for every minimum $x$ and an
ascent tree $\ascendtree(y)$ for every maximum $y$ of
$\terrain$. These are auxiliary structures that we later on will use
to efficiently map between vertices of $\terrain$ and leaves of
$\contourtree$.

We process a $\changeheight(v,r)$ operation as a continuous deformation of
$\terrain$ over time, that changes the height of $v$ to $r$. During this
continuous deformation the combinatorial structure of $\contourtree$ changes
only at discrete time instances, called \emph{events}. For simplicity, we
assume that no multiple saddle vertex is created during this deformation.

Edelsbrunner et al.~\cite{edelsbrunner:socg2008} showed that if $h$ is a smooth function, then the
topology of $\contourtree$ changes when (i) a critical point $u$ becomes
degenerate (i.e. the Hessian at $u$ becomes singular), or (ii)
$\height(u)=\height(v)$ for two saddle points $u$ and $v$ and both $u$ and $v$
lie on the same contour. In either case, one of the edges of $\contourtree$ is
degenerate in the sense that the interval $[h(u), h(v)]$ is a single point.
In our setting, where $h$ is a piecewise-linear function, (i)
corresponds to two adjacent vertices $u$ and $v$ with $\height(u)=\height(v)$
and one of them being a saddle and the other an extremal vertex; (ii)
corresponds to two saddle vertices lying on the same contour. The former is
called a \emph{birth} or a \emph{death} event and the latter is called an
\emph{interchange} event.

Besides these two events, there is another event in the
piecewise-linear case, namely, a critical point shifts from one vertex
to its neighbor --- no new critical points are created, none is
destroyed, and the topology of $\contourtree$ does not change. Only
the label of a node in $\contourtree$ changes. We refer to this event
as the \emph{shift} event.  Finally, the ascent and descent trees also
change at certain time instances, when an oriented edge $(u,v)$ will
be replaced with another edge $(u,w)$ or with the edge $(v, u)$.  The
birth, death, shift, and auxiliary events will be referred to as
\emph{local} events because they only occur when the height of two
adjacent vertices becomes equal.

Before we describe the events in detail we introduce some
notation. During the deformation, we use $\heighttime{t} : \mesh
\rightarrow \mathbb{R}$ to denote the height function at time
$t$. Note that during $\changeheight(v, r)$, $\heighttime{t}$ changes
only for points in $\star(v)$. If an event occurs at time $t$, then we
refer to $t^-$ (resp. $t^+$) as the time $t-\varepsilon$ (resp.
$t+\varepsilon)$ for some arbitrarily small $\varepsilon > 0$. We will
use $\contour_{\alpha\beta}^-$ ($\contour_{\alpha\beta}^+$) to refer
to a contour that retracts to the interior of edge $(\alpha,\beta)$ in
$\contourtree$ at time $t^-$ ($t^+$). A contour through a critical
vertex $\beta$ at time $t^-$ ($t^+$) will simply be denoted
$\contour_{\beta}^-$ ($\contour_{\beta}^+$). If $\beta$ is a saddle
vertex then $\contour_{\beta}^-$ ($\contour_{\beta}^+$) consists of
two simple polygonal cycles $\cycle_{\alpha\beta}^-$
($\cycle_{\alpha\beta}^+$) and $\cycle_{\bua\beta}^-$
($\cycle_{\bua\beta}^+$), belonging to the equivalence classes of
contours corresponding to the edges $(\alpha,\beta)$ and $(\bua,
\beta)$ of $\contourtree$, respectively; $\beta$ is the only common
point of $\cycle_{\alpha\beta}^-$ ($\cycle_{\alpha\beta}^+$) and
$\cycle_{\bua\beta}⁻$ ($\cycle_{\bua\beta}^+$).

\subsection{Local events} \label{sec:local-events}
Suppose a local event occurs at time $t_0$ at which
$\heighttime{t_0}(v) = \heighttime{t_0}(u)$ where $u$ is a neighbor of
$v$. For simplicity, we assume that $\heighttime{t_0^-}(v) <
\heighttime{t_0^-}(u)$ and $\heighttime{t_0^+}(v) >
\heighttime{t_0^+}(u)$, i.e., the height of $v$ is being raised.  The
other case when $\heighttime{t_0^-}(v) > \heighttime{t_0^-}(u)$ is
symmetric, simply reverse the direction of time. In the following
sections we describe in detail the changes that occur to
$\contourtree$ during the three kinds of local events. We assume the
interval $[t_0^-,t_0^+]$ is sufficiently small so that there is no
vertex $w \neq u,v$ whose height lies between those of $u$ and $v$
during this interval.

\paragraph{Auxiliary event} An auxiliary event occurs at $t_0$ if $(v,u)$ is
an edge in either an ascent tree a descent tree.

First, suppose $(v,u)$ is an edge of an ascent tree at $t_0^-$. We remove the
edge $(v,u)$. If $v$ becomes a maximum vertex at $t_0^+$, then $v$ becomes the
root of an ascent tree; otherwise, we choose another vertex $w$ from
$\upperlink(v)$ and add the edge $(v,w)$. If $u$ was a maximum at $t_0^-$, i.e.,
$u$ is the root of an ascent tree, then we add the edge $(u,v)$.

Next, suppose $(v,u)$ is an edge in a descent tree at $t_0^-$. We delete the edge
$(v,u)$. If $u$ becomes a minimum vertex at $t_0^+$, $u$ becomes the root of a
descent tree; otherwise we choose a vertex $w$ from $\lowerlink(u)$ and add the
edge $(w,u)$. If $v$ was a minimum vertex at $t_0^-$, it is no longer a minimum
at $t_0^+$ and we add the edge $(u,v)$.

\paragraph{Shift event} A shift event occurs at $t_0$ if one of $u$ and
$v$, say $v$, was a critical vertex and other vertex, $u$, was a regular vertex
at $t_0^-$, and the critical vertex shifts from $v$ to $u$ at $t_0^+$. This
event does not cause any change in the topology of $\contourtree$ but the node
of $\contourtree$ that was labeled $v$ changes its label to $u$. If
$\heighttime{t_0^+}(v) > \heighttime{t_0^-}(v)$, i.e., $v$ is being raised
then at time $t_0^+$, $\retract(v)$ lies on an edge of $\contourtree$ whose
lower endpoint is $u$.

\paragraph{Birth/death event} 
A \emph{birth event} occurs at time $t_0$ if both $u$ and $v$ were
regular vertices at $t_0^-$, and they become critical vertices at
$t_0^+$. A \emph{death event} occurs when both $u$ and $v$ were critical
vertices at $t_0^-$ and become regular vertices at $t_0^+$. See
Figure~\ref{fig:birthdeath} for the change in the topology of $\contourtree$.
We now describe in detail how $\contourtree$ changes at a birth or a death
event.

\begin{figure} 
\centering
\includegraphics[width=0.25\textwidth, page=7]{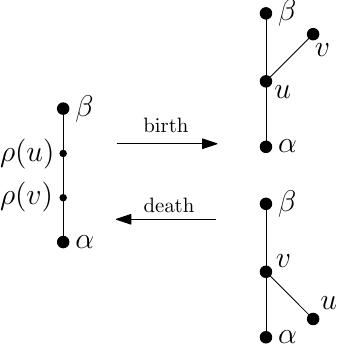}
\caption{Illustration of the change in the topology of the contour tree in birth and death events}
\label{fig:birthdeath}
\end{figure}

\textit{\textbf{Birth event:}}
If $v$ is being raised, then there are two possibilities: (i) $v$ becomes a
negative saddle and $u$ a minimum, or (ii) $v$ becomes a maximum and
$u$ a positive saddle. Suppose $\retract(u), \retract(v)$ lie on the
edge $(\alpha, \beta)$ of $\contourtree$. Then we split $(\alpha,
\beta)$ into two edges by adding a node corresponding to the new
saddle and creating a new edge incident on this node whose other
endpoint is a leaf. In case (i), $v$ is the node added on $(\alpha,
\beta)$ and $u$ is a new leaf, and in (ii) $u$ is the node on the edge
$(\alpha, \beta)$ and $v$ is the new leaf. 

\textit{\textbf{Death event:}}
Again, if $v$ is being raised, then there are two possibilities: (i) $v$ is a
minimum and $u$ a negative saddle at $t_0^-$, and (ii) $v$ is a positive saddle
and $u$ a maximum at $t_0^-$. In either case the edge $(u,v)$ disappears from
$\contourtree$ at $t_0^+$. Two edges incident on this degree 2 vertex merge
into a single edge.

\paragraph{Proof of correctness}
We now prove that whenever a local event occurs, then either the topology of
$\contourtree$ changes because of a birth/death event or the label of a node in
$\contourtree$ changes because of a shift event. There are three cases
depending on whether $u$ is an extremal, saddle or a regular vertex at time
$t_0^-$.
\\\\
(i) \emph{$v$ is an extremal vertex at time $t_0^-$.}  
Since $\heighttime{h_0^-}(v) < \heighttime{h_0^-}(u)$, $v$ cannot be a maximum
at time $t_0^-$, so assume that $v$ is a minimum at $t_0^-$. We observe that
$u$ also cannot be a maximum vertex because for any $w \in \link(v) \cup
\link(u)$, $\heighttime{t_0^-}(w) > \heighttime{t_0^-}(v)$, and therefore
$\heighttime{t_0^-}(w) > \heighttime{t_0^-}(u)$.  If $u$ is a \emph{regular
vertex} at time $t_0^-$, then $v$ is the only vertex in $\lowerlink(u)$ at
$t_0^-$. At time $t_0$ a shift event occurs that shifts the minimum from $v$ to
$u$. At $t_0^+$, $v$ is a regular vertex with $\lowerlink(v)=u$. 

If $u$ is a saddle vertex at time $t_0^-$, then at time $t_0$ a death event
occurs such that both $v$ and $u$ are regular vertices at time $t_0^+$. See
Figure~\ref{fig:local-event} (a) 
\\\\
(ii) \emph{$v$ is a regular vertex at time $t_0^-$.} 
If $u$ is the only vertex in $\upperlink(v)$, then $v$ becomes a maximum
vertex. If $u$ is a maximum at $t_0^-$, then a shift event occurs at $t_0$ and
$u$ becomes a regular vertex at $t_0^+$, and if $u$ is a regular vertex at
$t_0^-$ then a birth event occurs at $t_0$ and $u$ becomes a saddle vertex at
$t_0^+$; see Figure~\ref{fig:local-event} (b) (Note that $u$ cannot be a saddle
vertex at $t_0^-$ because otherwise $u$ becomes a multiple saddle at $t_0^+$.)

If $\upperlink(v)$ contains multiple vertices and $u$ is an endpoint of
$\upperlink(v)$ (degree of $u$ in $\upperlink(v)$ is 1), then $v$ simply
remains a regular vertex at $t_0^+$ and $u$ switches from $\upperlink(v)$ to
$\lowerlink(v)$. This does not cause $\contourtree$ to change. 

Finally, if $u$ is a middle vertex in $\upperlink(v)$ (degree of $u$ in
$\upperlink(v)$ is 2), $v$ becomes a saddle vertex at time $t_0^+$. If $u$ is
a regular vertex at time $t_0^-$, a birth event occurs at $t_0$ that creates a
minimum vertex at $u$ at $t_0^+$. If $u$ is a saddle vertex at time $t_0^-$,
then a shift event occurs at $t_0$ and becomes a regular vertex at $t_0^+$.
\\\\
(iii) \emph{$v$ is a saddle vertex at time}
$t_0^-$.  Note that since we assume that no multiple saddles can be created
during deformation, $u$ can not be the middle vertex of $\upperlink(v)$ at time
$t_0^-$.  If $u$ is the only vertex in $\upperlink(v)$, then $v$ becomes a
regular vertex. If $u$ is a regular vertex at $t_0^-$, then a shift event
occurs at $t_0$ and $u$ becomes a saddle vertex at $t_0^+$, and if $u$ is a
maximum vertex at $t_0^-$ then a death event occurs at $t_0$ and $u$ becomes a
regular vertex at $t_0^+$; see Figure~\ref{fig:local-event} (c).  (Note that $u$
cannot be a saddle vertex at $t_0^-$ because otherwise $u$  becomes a multiple
saddle at $t_0^+$.)

Finally, If $u$ is an endpoint of $\upperlink(v)$
neither $v$ nor $u$ changes vertex type. 

\begin{figure} 
\centering
\begin{tabular}{|>{\centering\arraybackslash}m{0.5cm}|>{\centering\arraybackslash}m{7cm}|>{\centering\arraybackslash}m{7cm}|}
\hline
& \emph{Shift event} & \emph{Birth / Death event} \\
\hline
\hline
(a)
& \includegraphics[width=0.42\textwidth, page=1]{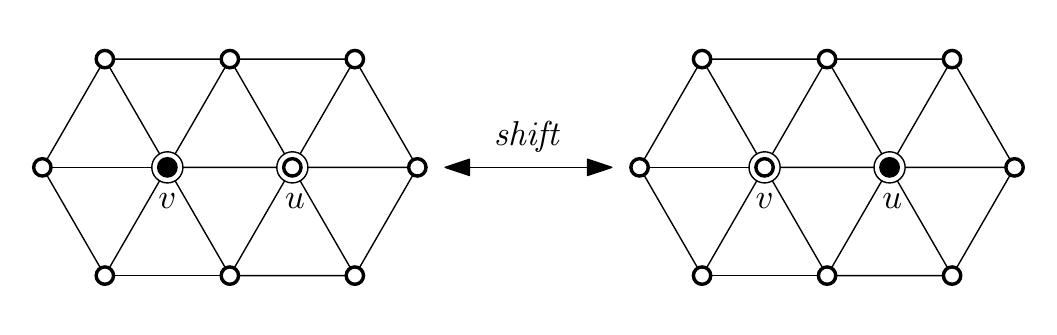} 
& \includegraphics[width=0.42\textwidth, page=2]{figures/local-event} \\
\hline
(b)
& \includegraphics[width=0.42\textwidth, page=3]{figures/local-event} 
& \includegraphics[width=0.42\textwidth, page=4]{figures/local-event} \\
\hline
(c)
& \includegraphics[width=0.42\textwidth, page=5]{figures/local-event} 
& \includegraphics[width=0.42\textwidth, page=6]{figures/local-event} \\
\hline
\end{tabular}
\caption{Illustration of local events. (a) is when $v$ is a minimum
  vertex. (b) is when $v$ is a regular vertex. (c) is when $v$ is a
  saddle vertex. Hollow vertices has height higher than both $u$ and
  $v$ and similarly filled vertices has lower height. In all examples
  $v$ is raised i.e.  $\height(v) < \height(u)$ before the event and
  $\height(v) > \height(u)$ after the event.}
\label{fig:local-event}
\end{figure}

\subsection{Interchange events} \label{sec:interchange-events}
An interchange event occurs at time $t_0$ if there are two saddle
vertices $\alpha, \beta$ such that $\heighttime{t_0}(\alpha) =
\heighttime{t_0}(\beta)$ and both $\alpha$ and $\beta$ lie on the same
contour, i.e., $\retract(\alpha) = \retract(\beta)$ at time
$t_0$. There are four cases depending on whether $\alpha$ and $\beta$
are positive or negative saddles.  Suppose $\heighttime{t_0^-}(\alpha)
< \heighttime{t_0^-}(\beta)$.  Without loss of generality assume that
$\alpha$ is a negative saddle. The other case, when $\alpha$ is
positive can be reduced to this case by reversing the $z$-axis and/or
time axis; see below. Then there are two cases: (i) $\alpha$ is
negative and $\beta$ is positive , and (ii) both $\alpha$ and $\beta$
are negative saddles. We refer to them as \emph{mixed} and
\emph{negative} interchange events.  We describe these events and
their effect on $\contourtree$ in detail in the following sections,
but first we introduce some notation.

At time $t_0^-$ all contours in the equivalence class $(\alpha$,
$\beta)$ are combinatorially identical because no vertex of $\mesh$
retracts to the interior of the edge ($\alpha$, $\beta$) of
$\contourtree$. With a small abuse of notation, we will therefore
simply refer to all contours in ($\alpha$, $\beta$) as the contour
$\contour^-$ without specifying a certain level of the
contour. Similarly, at time $t_0^+$ all contours in ($\beta$,
$\alpha$) are combinatorially identical, we refer to these as
$\contour^+$. We label the vertices of $\contour^-$ ($\contour^+$)
that lie on edges incident to $\alpha$ (resp. $\beta$) with $\alpha$
(resp.  $\beta$). 

\subsubsection{Mixed interchange event}
We first consider the case when $\alpha$ is negative and $\beta$ is
positive. We assume that the edge $(\alpha, \beta)$ is blue at
$t_0^-$, so both $\alpha$ and $\beta$ are blue at $t_0^-$. The case
when $(\alpha, \beta)$ is red reduces to this case by reversing the
direction of the $z$-axis. Let $\ada$ and $\adb$ be the two down
neighbors of $\alpha$ and $\bua,\bub$ the up neighbors of $\beta$ at
time $t_0^-$.  Both $(\ada,\alpha)$ and $(\adb,\alpha)$ are blue
edges.  Assume without loss of generality that the edge $(\beta,
\bua)$ is blue and $(\beta, \bub)$ is red i.e. the up-contours
$\contour_{\beta\bub}^-$ of $\beta$ is red and the up-contour
$\contour_{\beta\bua}^-$ of $\beta$ is blue. Refer to Figure
\ref{fig:mixed-transitions}.

Since $\alpha$ is a negative saddle and $\beta$ is positive, the
vertices of $\contour^-$ labeled with $\alpha$ form two intervals that
are connected and non-adjacent in $\contour^-$. The same is true for
vertices labeled $\beta$ ($\contour^-$ intersects both components of
$\upperlink(\alpha)$ and $\lowerlink(\beta)$). Since $\beta$ is the
only vertex between $\contour^-$ and the up-contours
$\contour_{\beta\bua}^-$, $\contour_{\beta\bub}^-$ of $\beta$, all
vertices of $\contour^-$ are either interior to an edge in
$\seqedge(\contour_{\beta\bua}^-)$, labelled with $\beta$ or interior
to an edge in $\seqedge(\contour_{\beta\bub}^-)$. We mark the portion
of $\contour^-$ that intersects $\seqedge(\contour_{\beta\bub}^-)$ red
and the portion that intersects $\seqedge(\contour_{\beta\bua}^-)$
blue according to the color of the contours $\contour_{\beta\bub}^-$
and $\contour_{\beta\bua}^-$, respectively. Refer to Figure
\ref{fig:cycle_minus}. Similarly, let $\contour_{\ada\alpha}^-$ and
$\contour_{\adb\alpha}^-$ be the down-contours of $\alpha$, then all
vertices of $\contour^-$ are either interior to an edge in
$\seqedge(\contour_{\ada\alpha}^-)$, labelled with $\alpha$ or
interior to an edge in $\seqedge(\contour_{\adb\alpha}^-)$. There are
three types of mixed interchange events depending on the relative
positions of the vertices marked $\alpha$ and those marked $\beta$ in
$\contour^-$. See Figure~\ref{fig:cycle_minus}.

\begin{figure}
\centering
\begin{tabular}{ccc}
\includegraphics[width=0.31\textwidth,page=16]{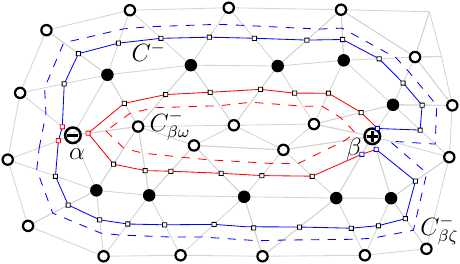} &
\includegraphics[width=0.31\textwidth,page=13]{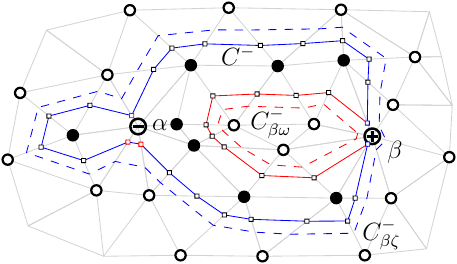} &
\includegraphics[width=0.31\textwidth,page=13]{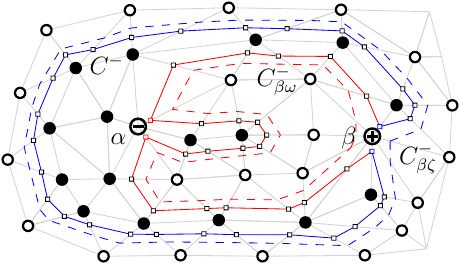} \\
(a) & (b) & (c)
\end{tabular}
\caption{Illustration of the possible colorings of $\contour^-$ for
  mixed interchange events. Showing the negative saddle $\alpha$ and
  possitive saddle $\beta$. Red (blue) contour vertices represent the
  vertices marked with $\alpha$ (resp. $\beta$). Hollow vertices of $\mesh$ has height higher than both
  $\alpha$ and $\beta$ and similarly filled vertices has lower height.
  (a) Contour vertices marked $\alpha$ and $\beta$ are interleaved
  along $\contour^-$.  (b) Contour vertices marked $\alpha$ lie in the
  blue part of $\contour^-$. (c) Contour vertices marked $\beta$ lie
  in the red part of $\contour^-$.}
\label{fig:cycle_minus}
\end{figure}

\setlist[description]{leftmargin=*}
\begin{description}
\item[(i)] A \emph{sign-interchange} when vertices marked $\alpha$ and
  those marked $\beta$ in $\contour^-$ are interleaved (vertices
  marked $\alpha$ intersect both $\seqedge(\contour_{\beta\bua}^-)$
  and $\seqedge(\contour_{\beta\bub}^-)$).
\item[(ii)] A \emph{blue} event when vertices marked $\alpha$ lie in
  the blue portion of $\contour^-$ (vertices marked $\alpha$ intersect
  $\seqedge(\contour_{\beta\bua}^-)$).
\item[(iii)] A \emph{red} event when vertices marked $\alpha$ lie in
  the red portion of $\contour^-$ (vertices marked $\alpha$ intersect
  $\seqedge(\contour_{\beta\bub}^-)$).
\end{description}

In case (ii) and (iii), without loss of generality assume that the
vertices labeled $\beta$ intersect
$\seqedge(\contour_{\adb\alpha}^-)$. The following lemma then
characterizes the change of $\contourtree$ at a mixed interchange
event.

\begin{lemma}\label{lemma:mixed}
Assuming that the edge $(\alpha,\beta)$ is blue at $t_0^-$ and $\alpha$ lying
below $\beta$ at $t_0^-$, the following change occurs in $\contourtree$ at
$t_0$:

\begin{description}
\item[(i)] At a sign-interchange event, the topology of $\contourtree$
  does not change. The only change is that the label $\alpha$ and
  $\beta$ of $\contourtree$ get swapped, so $\alpha$ becomes a
  positive (resp. negative) saddle at $t_0^+$. See
  Figure~\ref{fig:mixed-transitions} (a).
\item[(ii)] At a blue event, the signs of $\alpha$ and $\beta$ do not
  change, and the color of $(\alpha, \beta)$ remains blue. At time
  $t_0^+$, $\adb$ is the down neighbor of $\beta$, $\alpha$ and $\bub$
  are the up neighbors of $\beta$, $\ada$ and $\beta$ are down
  neighbors of $\alpha$, and $\bua$ is the up neighbor of
  $\alpha$. See Figure~\ref{fig:mixed-transitions} (b).
\item[(iii)] At a red event, the signs of $\alpha$ and $\beta$ do not
  change but the edge $(\alpha, \beta)$ becomes red and so does the
  saddle $\alpha$; $\beta$ remains a blue positive
  saddle. Furthermore, $\adb$ is down neighbor of $\beta$, $\alpha$
  and $\bua$ are the up neighbors of $\beta$, $\ada$ and $\beta$ are
  down neighbors of $\alpha$, and $\bub$ is the up neighbor of
  $\alpha$. See Figure~\ref{fig:mixed-transitions} (c).
\end{description}
\end{lemma}
\begin{proof}
We prove the case (i) in detail and sketch the proof for the other two cases,
as the argument is similar to case (i).

Consider the sign-interchange event. Let $\contour_{\ada\alpha}^-$,
$\contour_{\adb\alpha}^-$ and $\contour_{\beta\bua}^-$,
$\contour_{\beta\bub}^-$ be the the down-contours of $\alpha$ and the
up-contours of $\beta$ at time $t_0^-$, respectively. Let
$\theta_{down}=\heighttime{t_0^-}(\alpha)-\epsilon$ and
$\theta_{up}=\heighttime{t_0^-}(\beta)+\epsilon$ be the level of the
down-contours and up-contours, respectively. Assume that $\theta_{up}
> \heighttime{t_0^+}(\alpha)$.

First we fix the time $t_0^-$ and consider the contour $\contour^-$ as
we decrease the level towards $\heighttime{t_0^-}(\alpha)$. The
vertices of $\contour^-$ marked $\alpha$ converge to $\alpha$ of
$\terrain$, and we obtain $\contour_{\alpha}^-$. Since the vertices
marked $\alpha$ and $\beta$ in $\contour^-$ are interleaved, each
$\cycle_{\ada\alpha}^-$ and $\cycle_{\adb\alpha}^-$ contains an
interval where the vertices are marked $\beta$ and therefore so does
$\contour_{\ada\alpha}^-$ and $\contour_{\adb\alpha}^-$. Similarly,
$\contour_{\beta\bua}^-$ and $\contour_{\beta\bub}^-$ each contain an
interval where the vertices are marked $\alpha$. Refer to
Figure~\ref{fig:sign-interchange-proof}(a).

Next we fix the height $\theta_{down}$ and $\theta_{up}$ and move
forward in time. As time moves towards $t_0^+$,
$\contour_{\ada\alpha}^-$, $\contour_{\adb\alpha}^-$,
$\contour_{\beta\bua}^-$ and $\contour_{\beta\bub}^-$ continuously
deforms but no topological changes occur to the contours. Let
$\contourdeform_{\ada\alpha}: \mathbb{R}^2\times[t_0^-;t_0^+]
\rightarrow \mathbb{R}^2$ represent the continuous deformation of
$\contour_{\ada\alpha}^-$. For brevity of description we will simply
refer to $\contourdeform_{\ada\alpha}(\contour_{\ada\alpha}^-,t_0^+)$
as $\contourdeform(\contour_{\ada\alpha}^-)$. Similarly for
$\contourdeform(\contour_{\adb\alpha}^-)$,
$\contourdeform(\contour_{\beta\bua}^-)$ and
$\contourdeform(\contour_{\beta\bub}^-)$. Note that
$\seqedge(\contour_{\ada\alpha}^-) =
\seqedge(\contourdeform(\contour_{\ada\alpha}^-))$,
$\seqedge(\contour_{\adb\alpha}^-) =
\seqedge(\contourdeform(\contour_{\adb\alpha}^-))$,
$\seqedge(\contour_{\beta\bua}^-) =
\seqedge(\contourdeform(\contour_{\beta\bua}^-))$ and
$\seqedge(\contour_{\beta\bub}^-) =
\seqedge(\contourdeform(\contour_{\beta\bub}^-))$. Refer to
Figure~\ref{fig:sign-interchange-proof}(b).

Now we fix the time to $t_0^+$ and increase the height from
$\theta_{down}$ and monitor how the contours
$\contourdeform(\contour_{\ada\alpha}^-)$ and
$\contourdeform(\contour_{\adb\alpha}^-)$ deform. Recall that
$\heighttime{t_0^+}(\beta) < \heighttime{t_0^+}(\alpha)$, so as we
increase the height, we first encounter $\beta$ and the vertices of
$\contourdeform(\contour_{\ada\alpha}^-),
\contourdeform(\contour_{\adb\alpha}^-)$ marked $\beta$ converge to
the vertex $\beta$ of $\mesh$. Hence, $\beta$ is now a negative
saddle. Since both $\contourdeform(\contour_{\ada\alpha}^-),
\contourdeform(\contour_{\adb\alpha}^-)$ are blue, so is the saddle
$\beta$ at $t_0^+$. The up-contour $\contour^+$ of $\beta$ at time
$t_0^+$ is therefore a blue contour such that $\seqedge(\contour^+)$
contains all edges of
$\seqedge(\contourdeform(\contour_{\ada\alpha}^-)) \cup
\seqedge(\contourdeform(\contour_{\adb\alpha}^-))$ that are not
incident on $\beta$. If we continue to increase the height, the
vertices of $\contour^+$ marked $\alpha$ converge to $\alpha$ as we
reach $\heighttime{t_0^+}(\alpha)$, and it splits into two contours at
$\alpha$. Since $\contour^+$ was a blue contour, $\alpha$ is now a
blue positive saddle. The up-contours of $\alpha$ at time $t_0^+$ will
be $\contourdeform(\contour_{\beta\bua}^-)$ and
$\contourdeform(\contour_{\beta\bub}^-)$, so $\bua$ and $\bub$ will be
up-neighbors of $\alpha$. Refer to
Figure~\ref{fig:sign-interchange-proof}(b).This completes the proof of
the first part of the lemma.

Next if the vertices marked $\alpha$ and $\beta$ are not interleaved
in $\contour^-$, then by our assumption the vertices marked $\beta$
only intersect
$\seqedge(\contourdeform(\contour_{\adb\alpha}^-))$. Hence, as we
increase the height from $\theta_{down}$ to
$\heighttime{t_0^+}(\beta)$, the vertices marked $\beta$ in
$\contourdeform(\contour_{\adb\alpha}^-)$ converge to the vertex
$\beta$ of $\mesh$ and $\contourdeform(\contour_{\adb\alpha}^-)$
splits into two contours at $\beta$.  Note that as we increase the
height $\contourdeform(\contour_{\ada\alpha}^-)$ also deforms but does
not meet $\beta$, as it has no vertices marked $\beta$. Hence $\beta$
is a blue positive saddle, with $\adb$ as the down neighbor of
$\beta$. Let $\contour_\beta^+$ be the contour passing through $\beta$
at time $t_0^+$. 

If the blue event occurs, then vertices marked $\alpha$ are contained
in edges of $\seqedge(\contourdeform(\contour_{\beta\bua}^-))$. We can
divide $\contour_\beta^+$ into two cycles $\cycle_{\beta\bub}^+$,
$\cycle_{\beta\alpha}^+$ of vertices intersecting
$\seqedge(\contourdeform(\contour_{\beta\bub}^-))$ and
$\seqedge(\contourdeform(\contour_{\beta\bua}^-))$, respectively. As
we increase the height to $\heighttime{t_0^+}(\alpha)$, the vertices
marked $\alpha$ in $\contour^+$ and
$\contourdeform(\contour_{\ada\alpha}^-)$ converge to $\alpha$, and
the two contours merge into a single contour $\contour_\alpha^+$ at
$\alpha$. Since the color of both $\contourdeform(\contour^+)$ and
$\contourdeform(\contour_{\ada\alpha}^-)$ is blue, $\alpha$ remains a
blue negative saddle. The up-contour of $\alpha$ at time $t_0^+$ is
$\contourdeform(\contour_{\beta\bua}^-)$ and the contour
$\contourdeform(\contour_{\beta\bub}^-)$ remains an up-contour of
$\beta$. This proves part (ii) of the lemma.

The proof for the third case is symmetric and omitted from here.
\end{proof}

The above lemma characterizes the changes in the contour tree at a mixed
interchange event under the assumption that $\alpha$ was a blue negative saddle
at $t_0^-$. As mentioned above, the other cases can be reduced to the above
case. In particular, if at time $t_0^-$, $\alpha$ is a red negative saddle,
reverse the direction of the $z$-axis; if $\alpha$ is a blue positive saddle,
reverse the direction of time; and if $\alpha$ is red positive saddle, then
reverse the direction of time as well as that of the $z$-axis.

\begin{figure}
\centering
\begin{tabular}{ |>{\centering\arraybackslash}m{.4cm}|| >{\centering\arraybackslash}m{2.0cm}| >{\centering\arraybackslash}m{3.1cm}| >{\centering\arraybackslash}m{3.1cm}| >{\centering\arraybackslash}m{3.1cm} | >{\centering\arraybackslash}m{2.0cm}| }
\hline
& $\contourtree$ & \multicolumn{3}{c|}{Contours} & $\contourtree$ \\
\cline{2-6}
& $t_0^-$ & $t_0^-$ & $t_0$ & $t_0^+$ & $t_0^+$ \\
\hline 
(a) & 
    & \includegraphics[width=0.20\textwidth, page=1]{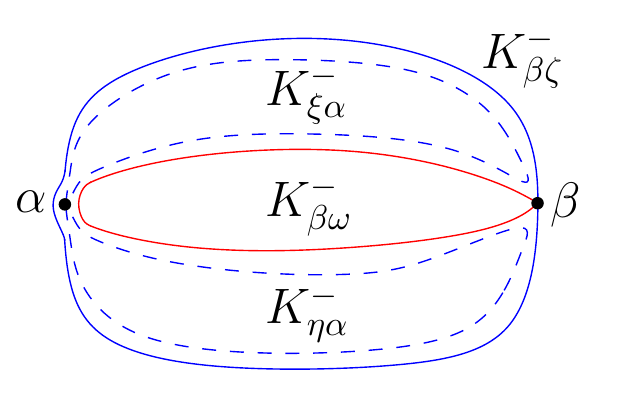} 
    & \includegraphics[width=0.20\textwidth, page=2]{figures/mixed-interchange-transitiona} 
    & \includegraphics[width=0.20\textwidth, page=3]{figures/mixed-interchange-transitiona} 
    & \includegraphics[width=0.10\textwidth, page=5]{figures/mixed-interchange-transitiona} \\
\cline{1-1} \cline{3-6} 
(b) & \includegraphics[width=0.10\textwidth, page=4]{figures/mixed-interchange-transitionb} 
    & \includegraphics[width=0.20\textwidth, page=1]{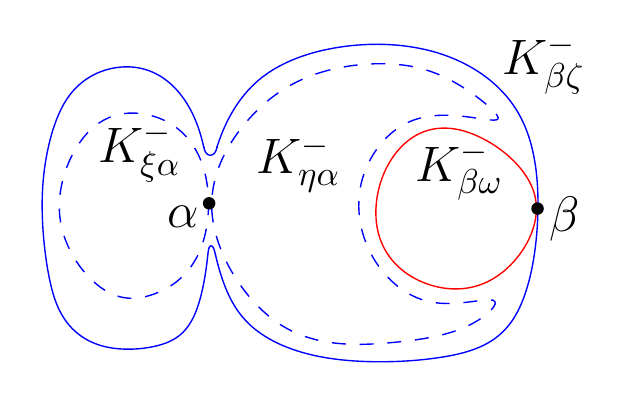} 
    & \includegraphics[width=0.20\textwidth, page=2]{figures/mixed-interchange-transitionb} 
    & \includegraphics[width=0.20\textwidth, page=3]{figures/mixed-interchange-transitionb} 
    & \includegraphics[width=0.10\textwidth, page=5]{figures/mixed-interchange-transitionb} \\
\cline{1-1} \cline{3-6} 
(c) & 
    & \includegraphics[width=0.20\textwidth, page=1]{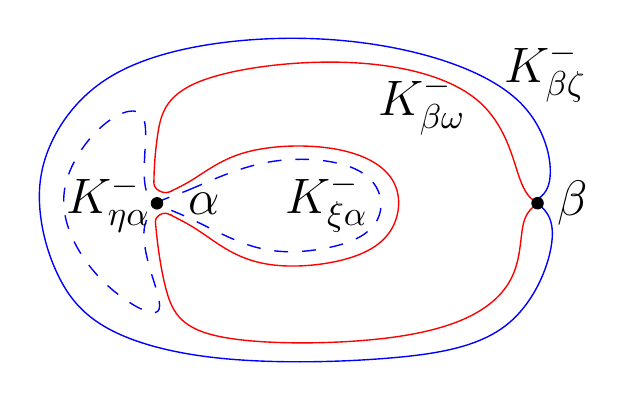} 
    & \includegraphics[width=0.20\textwidth, page=2]{figures/mixed-interchange-transitionc} 
    & \includegraphics[width=0.20\textwidth, page=3]{figures/mixed-interchange-transitionc} 
    & \includegraphics[width=0.10\textwidth, page=5]{figures/mixed-interchange-transitionc} \\
\hline 

\end{tabular}
\caption{Illustration of topological changes in $\levelset$ and
  contour tree transitions during mixed interchange events.  Dashed
  contour lines are contours at the height of $\alpha$ and solid
  contour lines are contours at the height of $\beta$. (a) Illustrates
  a sign change event. (b) Illustrates a blue event. (c) Illustrates a
  red event.}
\label{fig:mixed-transitions}
\end{figure}

\begin{figure} 
\centering
\begin{tabular}{cc}
\begin{tabular}{|>{\centering\arraybackslash}m{2cm}|>{\centering\arraybackslash}m{5cm}|}
\hline
Height & Contours \\
\hline
$\heighttime{t_0^-}(\beta)+\varepsilon$ &
\includegraphics[width=.3\textwidth, page=5]{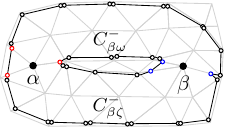} \\
\hline
$\heighttime{t_0^-}(\beta)$ &
\includegraphics[width=.3\textwidth, page=6]{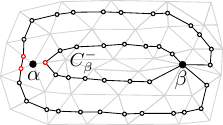} \\
\hline
$ < \heighttime{t_0^-}(\beta)$ $ > \heighttime{t_0^-}(\alpha)$ &
\includegraphics[width=.3\textwidth, page=7]{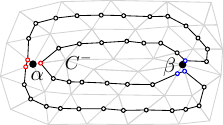} \\
\hline
$\heighttime{t_0^-}(\alpha)$ &
\includegraphics[width=.3\textwidth, page=8]{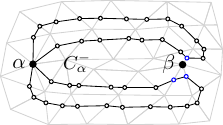} \\
\hline
$\heighttime{t_0^-}(\alpha)+\varepsilon$ &
\includegraphics[width=.3\textwidth, page=9]{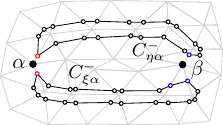} \\
\hline
\end{tabular}
& 
\begin{tabular}{|>{\centering\arraybackslash}m{2cm}|>{\centering\arraybackslash}m{5cm}|}
\hline
Height & Contours \\
\hline
$\heighttime{t_0^-}(\beta)+\varepsilon$ &
\includegraphics[width=.3\textwidth, page=10]{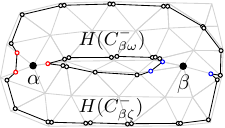} \\
\hline
$\heighttime{t_0^+}(\alpha)$ &
\includegraphics[width=.3\textwidth, page=11]{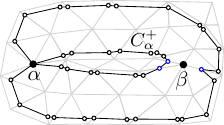} \\
\hline
$ < \heighttime{t_0^+}(\alpha)$ $ > \heighttime{t_0^+}(\beta)$ &
\includegraphics[width=.3\textwidth, page=12]{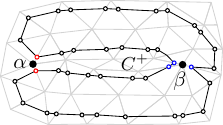} \\
\hline
$\heighttime{t_0^+}(\beta)$ $(\heighttime{t_0^-}(\beta))$ &
\includegraphics[width=.3\textwidth, page=13]{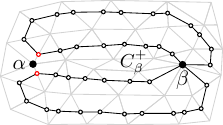} \\
\hline
$\heighttime{t_0^-}(\alpha)+\varepsilon$ &
\includegraphics[width=.3\textwidth, page=14]{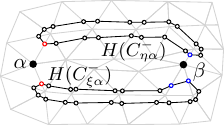} \\
\hline
\end{tabular} \\
(a) $t = t_0^-$ & (b) $t = t_0^+$
\end{tabular} 
\caption{Illustration of sign change event. The vertices marked $\alpha$ are
illustrated as red vertices and the vertices marked $\beta$ as blue vertices.}
\label{fig:sign-interchange-proof}
\end{figure}

\begin{figure} 
\centering
\begin{tabular}{cc}
\begin{tabular}{|>{\centering\arraybackslash}m{2cm}|>{\centering\arraybackslash}m{5cm}|}
\hline
Height & Contours \\
\hline
$\heighttime{t_0^-}(\beta)+\varepsilon$ &
\includegraphics[width=.3\textwidth, page=10]{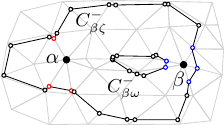} \\
\hline
$\heighttime{t_0^-}(\beta)$ &
\includegraphics[width=.3\textwidth, page=6]{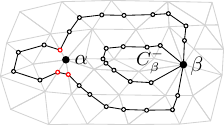} \\
\hline
$ < \heighttime{t_0^-}(\beta)$ $ > \heighttime{t_0^-}(\alpha)$ &
\includegraphics[width=.3\textwidth, page=5]{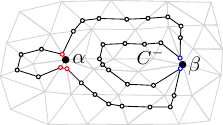} \\
\hline
$\heighttime{t_0^-}(\alpha)$ &
\includegraphics[width=.3\textwidth, page=4]{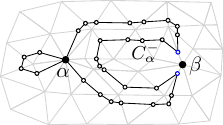} \\
\hline
$\heighttime{t_0^-}(\alpha)+\varepsilon$ &
\includegraphics[width=.3\textwidth, page=2]{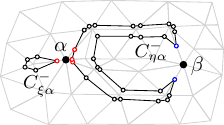} \\
\hline
\end{tabular}
&
\begin{tabular}{|>{\centering\arraybackslash}m{2cm}|>{\centering\arraybackslash}m{5cm}|}
\hline
Height & Contours \\
\hline
$\heighttime{t_0^-}(\beta)+\varepsilon$ &
\includegraphics[width=.3\textwidth, page=11]{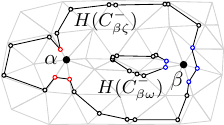} \\
\hline
$\heighttime{t_0^+}(\alpha)$ &
\includegraphics[width=.3\textwidth, page=9]{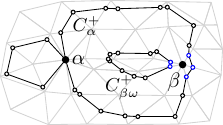} \\
\hline
$ < \heighttime{t_0^+}(\alpha)$ $ > \heighttime{t_0^+}(\beta)$ &
\includegraphics[width=.3\textwidth, page=8]{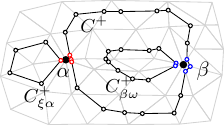} \\
\hline
$\heighttime{t_0^+}(\beta)$ $(\heighttime{t_0^-}(\beta))$ &
\includegraphics[width=.3\textwidth, page=7]{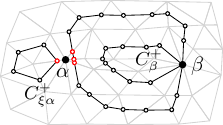} \\
\hline
$\heighttime{t_0^-}(\alpha)+\varepsilon$ &
\includegraphics[width=.3\textwidth, page=3]{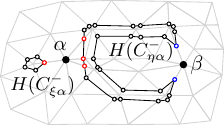} \\
\hline
\end{tabular} \\
(a) $t = t_0^-$ & (b) $t = t_0^+$
\end{tabular} 

\caption{Illustration of blue event. The vertices marked $\alpha$ are
illustrated as red vertices and the vertices marked $\beta$ as blue vertices.}
\label{fig:blue-proof}
\end{figure}

\subsubsection{Negative interchange event} \label{sec:neg-interchange}
Let $\ada, \adb$ be the two down neighbors of $\alpha$ at $t_0^-$, and let
$\bua$ be the other down neighbor of $\beta$ ($\alpha$ is a down neighbor
of $\beta$ at $t_0^-$). See Figure~\ref{fig:transitions}. The
change in topology of $\contourtree$ at a negative interchange event is similar
to performing a rotation at node $\beta$. That is, $\alpha$ becomes the upper
endpoint and $\beta$ the lower endpoint of the edge $(\alpha, \beta)$ at time
$t_0^+$, and one of the down subtrees of $\alpha$ (rooted in $\ada$ and $\adb$)
becomes a down subtree of $\beta$. Next we describe the change of
$\contourtree$ in more detail, and also argue why our analysis is correct.

Let $\contour_{\ada\alpha}^-$ and $\contour_{\adb\alpha}^-$ be the
down-contours of $\alpha$. Since both $\alpha$ and $\beta$ are
negative sadles, the vertices of $\contour^-$ labeled with $\alpha$
form two connected components, and the vertices labeled with $\beta$
form a single connected component ($\contour^-$ intersects both
components of $\upperlink(\alpha)$ but only a single component of
$\lowerlink(\beta)$). Since $\alpha$ is the only vertex between
$\contour^-$ and the down-contours of $\alpha$, all vertices of
$\contour^-$ are either interior to an edge in
$\seqedge(\contour_{\ada\alpha}^-)$, labelled with $\alpha$ or
interior to an edge in $\seqedge(\contour_{\adb\alpha}^-)$. Since we
assume the absence of multiple saddles, the vertices of $\contour^-$
labelled with $\beta$ can not intersect both
$\seqedge(\contour_{\ada\alpha}^-)$ and
$\seqedge(\contour_{\adb\alpha}^-)$.

\begin{lemma}\label{lemma:negative-interchange}
If the vertices of $\contour^-$ marked $\beta$ are interior to
$\seqedge(\contour_{\adb\alpha}^-)$, then at time $t_0^+$, $\adb$
becomes a down neighbor of $\beta$, $\bua$ the other down neighbor of
$\beta$, and $\ada$ and $\beta$ becomes the down neighbor of $\alpha$.
\end{lemma}
\begin{proof}
Let $\theta = \heighttime{t_0^-}(\alpha)-\epsilon$ be the level of
down-contours $\contour_{\ada\alpha}^-$ and $\contour_{\adb\alpha}^-$,
and let $\contour_{\bua\beta}^-$ be the contour at level $\theta$
along $(\bua, \beta)$. We fix the height at $\theta$ and go forward in
time. As time moves towards $t_0^+$, $\contour_{\ada\alpha}^-$,
$\contour_{\adb\alpha}^-$ and $\contour_{\bua\beta}^-$ continuously
deforms, but no topological changes occur to these contours. Let
$\contourdeform_{\ada\alpha}: \mathbb{R}^2\times[t_0^-;t_0^+]
\rightarrow \mathbb{R}^2$ represent the continuous deformation of
$\contour_{\ada\alpha}^-$. For brevity of description we will simply
refer to $\contourdeform_{\ada\alpha}(\contour_{\ada\alpha}^-,t_0^+)$
as $\contourdeform(\contour_{\ada\alpha}^-)$. Similarly for
$\contourdeform(\contour_{\adb\alpha}^-)$ and
$\contourdeform(\contour_{\bua\beta}^-)$. Assume that the vertices of
$\contour^-$ marked $\beta$ are interior to
$\seqedge(\contour_{\adb\alpha}^-)$ then $\contour_{\adb\alpha}^-$ and
$\contour_{\bua\beta}^-$ have vertices marked $\beta$; and
$\contour_{\ada\alpha}^-$, and $\contour_{\adb\alpha}^-$ have vertices
marked $\alpha$. Note that $\seqedge(\contour_{\ada\alpha}^-) =
\seqedge(\contourdeform(\contour_{\ada\alpha}^-))$,
$\seqedge(\contour_{\adb\alpha}^-) =
\seqedge(\contourdeform(\contour_{\adb\alpha}^-))$ and
$\seqedge(\contour_{\bua\beta}^-) =
\seqedge(\contourdeform(\contour_{\bua\beta}^-))$.

We now fix time at $t_0^+$. Recall that $\theta <
\heighttime{t_0^+}(\beta) < \heighttime{t_0^+}(\alpha)$. As we
increase the height from $\theta$ to $\heighttime{t_0^+}(\beta)$, the
vertices marked $\beta$ in $\contourdeform(\contour_{\adb\alpha}^-)$
and $\contourdeform(\contour_{\bua\beta}^-)$ converge to $\beta$, and
thus $\contourdeform(\contour_{\adb\alpha}^-)$ and
$\contourdeform(\contour_{\bua\beta}^-)$ merge into a single contour
$\contour^+$ at $\beta$.  Thus $\beta$ is a negative saddle and has
$\adb$ and $\bua$ as its down neighbors.

Next, we increase the height from $\heighttime{t_0^+}(\beta)$ to
$\heighttime{t_0^+}(\alpha)$.  The merged contour $\contour^+$,
contains vertices marked $\alpha$ and so does the contour
$\contourdeform(\contour_{\ada\alpha}^-)$. Hence as the height reaches
$\heighttime{t_0^+}(\alpha)$, the vertices marked $\alpha$ in
$\contourdeform(\contour_{\ada\alpha}^-)$ and $\contour^+$ converge
to $\alpha$.  Hence at time $t_0^+$, $\alpha$ is a negative saddle
with $\beta$ and $\ada$ as its two down neighbors. This completes the
proof of the lemma.
\end{proof}

Next, we analyze the change in
colors of $\alpha$ and $\beta$ at the above event. There are three
cases depending on the colors of $\alpha$ and $\beta$ at time
$t_0^-$.

First, assume that $\beta$ is blue at $t_0^-$, then $\alpha$ is also
blue at $t_0^-$, and so are the cycles $\cycle_{\alpha\beta}^-,
\cycle_{\bua\beta}^-, \cycle_{\ada\alpha}^-,
\cycle_{\adb\beta}^-$. See Figure~\ref{fig:transitions} (a). Moreover,
$\cycle_{\alpha\beta}^-, \cycle_{\bua\beta}^-$ lie in the exterior of
each other and so do $\cycle_{\ada\alpha}^-$ and
$\cycle_{\adb\alpha}^-$.  Consequently $\contour_\ada^-,
\contour_\adb^-, \contour_\bua^-$ lie in the exterior of each other,
and the same holds for $\contour_\ada^+, \contour_\adb^+$, and
$\contour_\bua^+$. We can now conclude that $\alpha$ and $\beta$
remain blue negative saddles at time $t_0^+$.

Next, assume that $\beta$ is red at $t_0^-$. Then one of $(\alpha,
\beta)$ or $(\bua, \beta)$ is red and the other is blue. We assume
that $(\alpha, \beta)$ is red; we will argue below that the other case
can be reduced to this case. In this case $\cycle_{\bua\beta}^-$ lies
in the interior of the cycle $\cycle_{\alpha\beta}^-$; see
Figure~\ref{fig:transitions} (b), (c). If $\alpha$ is red, then one of
$(\ada,\alpha)$ and $(\adb, \alpha)$ is red and the other is
blue. First assume that $(\adb, \alpha)$ is red at $t_0^-$. So
$\cycle_{\ada\alpha}^-$ lies in the interior of
$\cycle_{\adb\alpha}^-$: see Figure~\ref{fig:transitions} (b). Hence
$\contour_{\adb}^-$ contains both $\contour_\ada^-$ and
$\contour_\bua^-$ in its interior. Therefore at time $t_0$, the
contour $\cycle_{\alpha=\beta}$ passing through $\alpha$ and $\beta$
consists of three cycles $\cycle_{\ada\alpha}, \cycle_{\adb\alpha}$,
and $\cycle_{\bua\beta}$ --- $\cycle_{\adb\alpha}$ containing both
$\cycle_{\ada\alpha}$ and $\cycle_{\bua\beta}$ at $\alpha$ and
$\beta$, respectively. At time $t_0^+$, both $\alpha$ and $\beta$
remain red. See Figure~\ref{fig:transitions} (b).

Finally, consider the case when $(\ada, \alpha)$ is red and $(\adb, \alpha)$ is
blue. In this case, $\cycle_{\ada\alpha}^-$ contains $\cycle_{\adb\alpha}^-$ in
its interior. At time $t_0$, $\contour_{\alpha=\beta}$ consists of three cycles
$\cycle_{\ada\alpha}, \cycle_{\adb\alpha}$, and $\cycle_{\bua\beta}$ ---
$\cycle_{\ada\alpha}$ containing both $\cycle_{\adb\alpha}$ and
$\cycle_{\bua\beta}$ in its interior with $\cycle_{\ada\alpha}$ touching
$\cycle_{\adb\alpha}$ at $\alpha$ and $\cycle_{\adb\alpha}$ touching
$\cycle_{\bua\beta}$ at $\beta$. At time $t_0^+$, $\alpha$ remains red but
$\beta$ becomes blue. See Figure~\ref{fig:transitions} (c).

We conclude this case by noticing that the case when $\alpha$ is blue
can be reduced to the third case above by reversing the direction of
time. Refer to Figure~\ref{fig:transitions} (c).

\begin{figure}
\centering
\begin{tabular}{ |>{\centering\arraybackslash}m{.4cm}|| >{\centering\arraybackslash}m{2.0cm}| >{\centering\arraybackslash}m{3.1cm}| >{\centering\arraybackslash}m{3.1cm}| >{\centering\arraybackslash}m{3.1cm} | >{\centering\arraybackslash}m{2.0cm}| }
\hline
& $\contourtree$ & \multicolumn{3}{c|}{Contours} & $\contourtree$ \\
\cline{2-6}
& $t_0^-$ & $t_0^-$ & $t_0$ & $t_0^+$ & $t_0^+$ \\
\hline 
(a) & \includegraphics[width=0.10\textwidth, page=4]{figures/negative-interchange-transitiona} 
    & \includegraphics[width=0.20\textwidth, page=1]{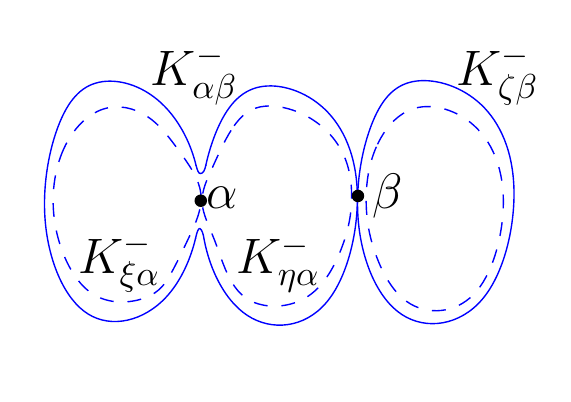} 
    & \includegraphics[width=0.20\textwidth, page=2]{figures/negative-interchange-transitiona}
    & \includegraphics[width=0.20\textwidth, page=3]{figures/negative-interchange-transitiona}
    & \includegraphics[width=0.10\textwidth, page=5]{figures/negative-interchange-transitiona} \\
\hline
(b) & \includegraphics[width=0.10\textwidth, page=4]{figures/negative-interchange-transitionb} 
    & \includegraphics[width=0.20\textwidth, page=1]{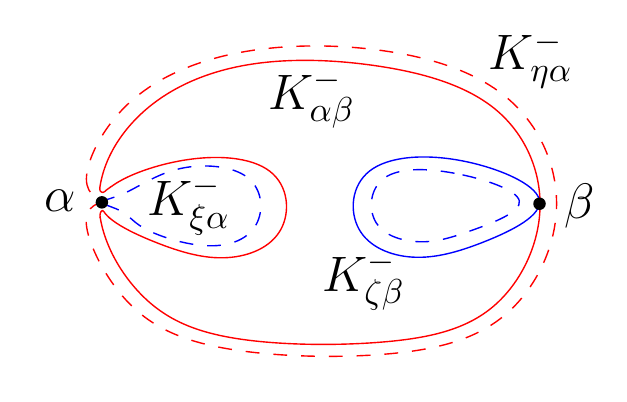} 
    & \includegraphics[width=0.20\textwidth, page=2]{figures/negative-interchange-transitionb} 
    & \includegraphics[width=0.20\textwidth, page=3]{figures/negative-interchange-transitionb}
    & \includegraphics[width=0.10\textwidth, page=5]{figures/negative-interchange-transitionb} \\
\hline 
(c) & \includegraphics[width=0.10\textwidth, page=4]{figures/negative-interchange-transitionc} 
    & \includegraphics[width=0.20\textwidth, page=1]{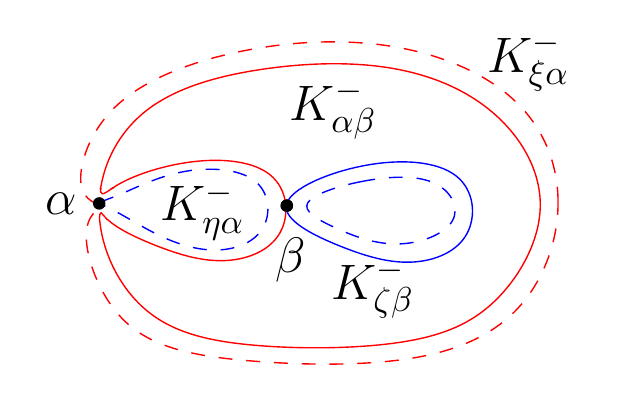} 
    & \includegraphics[width=0.20\textwidth, page=2]{figures/negative-interchange-transitionc}
    & \includegraphics[width=0.20\textwidth, page=3]{figures/negative-interchange-transitionc}
    & \includegraphics[width=0.10\textwidth, page=5]{figures/negative-interchange-transitionc} \\
\hline 

\end{tabular}
\caption{Illustration of topological changes in $\levelset$ and
  contour tree transitions during negative interchange events.  Dashed
  lines are contours at the height of $\alpha$ and solid lines are
  contours at the height of $\beta$. (a) is when $\cycle_{\beta}$ is
  blue. (b) and (c) are when $\cycle_{\beta}$ is red.}
\label{fig:transitions}
\end{figure}

%% file: data-structure.tex
\section{Event Handling} \label{sec:event-handling} 
In this section we describe how to efficiently maintain $\contourtree$
of $\terrain$ under the events described in the previous section. We
first describe the data structure used to represent the contour,
ascent and descent trees of $\terrain$. Second, we describe the
mechanisms that we use to repair this structure for each
event. Finally, we will describe how each event is detected during a
$\changeheight(v,r)$ operation.

\subsection{Data Structure}
Our data structure represents both the contour, ascent and descent
trees of $\terrain$ as \emph{link-cut trees}~\cite{st-dsdt-83} that
allow for efficiently maintaining a dynamic forrest $F$. Specifically,
for $v$ and $w$ in $F$ the $\textsc{link}(v,w)$ operation connects the
trees containing $v$ and $w$ by inserting the edge $(v,w)$, the
$\textsc{cut}(v,w)$ operation splits the tree containing $v$ and $w$
by removing the edge $(v,w)$, the $\textsc{evert}(v)$ operation makes
$v$ the root of the tree containing $v$, and the $\textsc{root}(v)$
operation queries for the root of the tree containing $v$. All
operations require $\OhOf(\log(n))$ time where $n$ is the size of the
forrest. Furthermore, if every edge of $F$ is associated with a cost,
we can evaluate functions on root-to-leaf paths finding e.g. the
minimum cost along the path in $\OhOf(\log(n))$ time. If the costs
along a root-to-leaf path in $F$ are decreasing we can also search for
edges with a given cost along the path in $\OhOf(\log(n))$
time.

We augment the vertices of $\terrain$ with pointers to their
corresponding instances in both the ascent and descent tree
forest. Also, we augment each critical vertex in $\terrain$ with a
pointer to its corresponding vertex in $\contourtree$. Finally, for
each vertex $v$ in $\terrain$ we conceptually order the vertices in
$\link(v)$ clockwise around $v$ and maintain a \emph{link pointer} to
the start and end vertex of each connected component of
$\lowerlink(v)$ and $\upperlink(v)$.

Our datastructure links the root $x$ of the descent tree
$\descendtree(x)$ and the root $y$ of the ascent tree $\ascendtree(y)$
to their corresponding leafs in $\contourtree$. Conceptually we have
the ascent and descent tree regions of $\terrain$ hanging of the leafs
of $\contourtree$. Refer to Figure~\ref{fig:subdivisions}. Consider
the operation $\textsc{FindEdge(v)}$ that given a vertex $v$ in
$\terrain$ returns the edge $(\alpha,\beta)$ of $\contourtree$
containing $\retract(v)$. Our structure easily supports this operation
in $\OhOf(\log(n))$ time. Assume that $v$ belongs to $\descendtree(x)$
and $\ascendtree(y)$ we compute $(\alpha,\beta)$ by retrieving leafs
$x$ and $y$ of $\contourtree$, rooting $\contourtree$ in $y$ and
finally searching for $(\alpha,\beta)$ on the root-to-leaf path from
$y$ to $x$.

\begin{figure}
  \centering{
    \begin{tabular}{ccc}
      \includegraphics[width=0.35\textwidth,page=2]{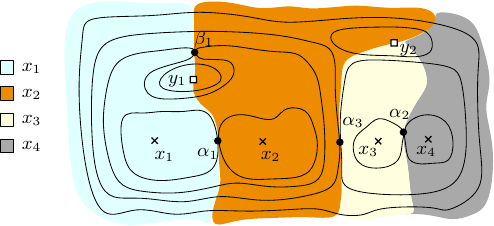} &
      \includegraphics[width=0.35\textwidth,page=3]{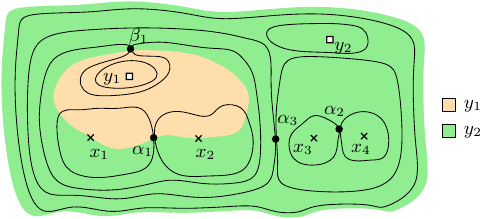} &
      \includegraphics[width=0.2\textwidth,page=4]{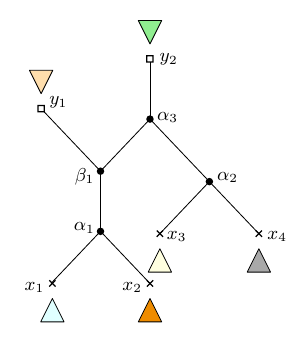} \\
      (a) & (b) & (c)
    \end{tabular}
  }
  \caption{(a) Illustrates the subdivision of the terrain represented by
    the descend tree forest. (b) Illustrates the subdivision of the
    terrain represented by the ascend tree forest. (c) Our data
    structure consists of the contour tree connected with descent trees
    at minimum vertices and with ascend trees at maximum vertices.}
  \label{fig:subdivisions}
\end{figure}

\subsection{Repair Mechanisms}
Based on the description in Section~\ref{sec:deformation} of events
and the changes they cause in the contour tree and ascent/descent
trees of $\terrain$, we will now go through each event type and show
how to efficiently repair our data structure. 

\paragraph{Auxiliary event}
Consider an auxiliary event that occurs on edge $(v,u)$ in descent
tree $\descendtree(x)$, we need to cut the subtree rooted in $u$ from
$\descendtree(x)$ and potentially link the subtree to some other
descent tree containing vertex $w$ in $\lowerlink(u)$. Since the
descent tree is represented as a dynamic tree and since $v$ points to
its instance in $\descendtree(x)$, this is easily done using
$\OhOf(\log(n))$ time. This is similar for auxiliary events that occur
on ascent trees. 

During an auxiliary event we might also need to update link pointers
in $v$ and $u$. The links change as described in the correctness proof
in Section \ref{sec:local-events} and since there are only a constant
number of link pointers, these can straightforwardly be updated in
$\OhOf(1)$ time using the clockwise ordering of $\link(v)$ and
$\link(u)$.

\paragraph{Shift event}
Consider a shift event that shifts a critical vertex between $v$ and
$u$. In this case we simply update our data structure by switching the
labels of $v$ and $u$ in $\contourtree$ and the ascent/descent trees
containing $u$ and $v$. This can be done in $\OhOf(1)$ time since each
vertex of $\terrain$ points to its instance in ascent/descent trees of
$\terrain$ and every critical vertex of $\terrain$ points to its
corresponding vertex in $\contourtree$ of $\terrain$.

\paragraph{Birth/death event}
Consider a birth event where both $u$ and $v$ were regular vertices
before the event and become critical vertices after the event. Before
the event both $\retract(u)$ and $\retract(v)$ lie on the same edge
$(\alpha,\beta)$ of $\contourtree$. The event splits $(\alpha,\beta)$
into two by adding a new saddle node to $\contourtree$ and creating an
edge incident on this node whose other endpoint is a leaf. Recall
Figure~\ref{fig:birthdeath}. Performing these changes on
$\contourtree$ is trivial once we know $(\alpha,\beta)$ and we can
simply use $\textsc{FindEdge(v)}$ to retrieve $(\alpha,\beta)$ in
$\OhOf(\log(n))$ time. Handling death events is trivial since $u$ and
$v$ are critical vertices with pointers to $\contourtree$, thus it
simply corresponds to removing a leaf edge of $\contourtree$.

\paragraph{Mixed interchange event}
Consider a mixed interchange event where a negative saddle vertex
$\alpha$ is raised above a positive saddle $\beta$ as described in
Section \ref{sec:interchange-events}. From Lemma~\ref{lemma:mixed}, we
get a complete description of the changes that occur to $\contourtree$
for each type of mixed interchange event, making these changes to our
representation of $\contourtree$ is trivial, so what remains is to
detect which kind of event that takes place at time $t_0$. Recall that
at time $t_0^-$ the vertices along $\contour^-$ marked with $\beta$
lie in two disjoint intervals. Removing these vertices from
$\contour^-$ divides the contour into two chains; One consisting of
vertices interior to edges in $\seqedge(\contour_{\beta\bua}^-)$ and
one consisting of vertices interior to edges in
$\seqedge(\contour_{\beta\bub}^-)$. From Section
\ref{sec:interchange-events} we get that a sign-interchange event
occurs at time $t_0$ if vertices of $\contour^-$ marked $\alpha$
belong both to edges in $\seqedge(\contour_{\beta\bua}^-)$ and
$\seqedge(\contour_{\beta\bub}^-)$; A blue event occurs at time $t_0$
if vertices of $\contour^-$ marked $\alpha$ belong to
$\seqedge(\contour_{\beta\bua}^-)$; A red event occurs at time $t_0$
if vertices of $\contour^-$ marked $\alpha$ belong to
$\seqedge(\contour_{\beta\bub}^-)$.

Note that an interval of $\contour^-$ marked $\alpha$ intersects edges
of $\terrain$ that are incident to vertices in the same connected
component of $\upperlink(\alpha)$. Therefore to determine whether this
interval belong to $\seqedge(\contour_{\beta\bua}^-)$ or
$\seqedge(\contour_{\beta\bub}^-)$, we can select a vertex $v_i$ in
the corresponding component of $\upperlink(\alpha)$, and determine
whether the edge $(\alpha,v_i)$ belong to
$\seqedge(\contour_{\beta\bua}^-)$ or
$\seqedge(\contour_{\beta\bub}^-)$. The edge $(\alpha,v_i)$ belongs to
$\seqedge(\contour_{\beta\bua}^-)$ if and only if there exists an
ascending path from $v_i$ that ends in a maximum vertex of $\terrain$
corresponding to a leaf node in the up subtree $\contourtree_{\bua}$
of $\contourtree$ rooted in $\bua$. Let $y_i$ be the maximum vertex of
$\terrain$ such that $v_i$ belongs to $\ascendtree(y_i)$, we check
that $y_i$ belongs to $\contourtree_{\bua}$, by simply determining
whether the path from $\alpha$ to $y_i$ in $\contourtree$ goes through
$\bua$.  Given $v_i$ our
algorithm retrieves $y_i$ as the root of the dynamic tree
$\ascendtree(y_i)$ in $\OhOf(\log(n))$ time. We then root
$\contourtree$ in $\alpha$ and test whether the first vertex on the
root-to-leaf path from $\alpha$ to $y_i$ is $\bua$. This is
straightforward to in $\OhOf(\log(n))$ time using standard dynamic
tree operations.

\paragraph{Negative Interchange Event}
Consider a negative interchange event where a negative saddle vertex
$\alpha$ is raised above a negative saddle $\beta$. As shown in
Section \ref{sec:interchange-events} there is only one kind of
topological change that can occur to $\contourtree$ during this event
and this change corresponds to performing a rotation at node
$\beta$. According to Lemma~\ref{lemma:negative-interchange} we can
determine which down subtree of $\alpha$ becomes the down subtree of
$\beta$ after the rotation, by determining whether vertices along
$\contour^-$ marked $\beta$ are within
$\seqedge(\contour_{\alpha\ada}^-)$ or
$\seqedge(\contour_{\alpha\adb}^-)$. We omit the details of this
algorithm since it is similar to the algorithm for handling mixed
interchange events except that we use the descent trees of $\terrain$
instead of the ascent trees.

\subsection{Event Detection}
Let $X$ contain the set of vertices in $\upperlink(v)$ as well as the
upper node of every edge $e$ in $\contourtree$ whose lower node is $v$
if $v$ is a saddle vertex. From the event description in
Section~\ref{sec:deformation}, we know that as we raise $v$ the next
event happens as $h(v)$ becomes equal to $h(u)$ for some $u$ in
$X$. To detect the next event, we maintain a priority queue of
vertices in $X$ with their height as priority, and simply retrieve the
minimum priority event. If $v$ and $u$ are saddle vertices then an
interchange event occurs otherwise a local event occurs. As described
in Section \ref{sec:local-events}, we distinguish between shift, birth
and death events by looking at how $u$ and $v$ change from regular to
critical vertex or vice versa during the event. Both before and after
an event, we can determine whether $u$ and $v$ are critical or regular
in $\OhOf(1)$ time using their link pointers. As also described in
Section \ref{sec:local-events}, an auxiliary event can be detected by
simply examining whether $(v,u)$ is in either the ascent tree pointed
to by $v$ or the descent tree pointed to by $u$.

If an interchange event occurs the endpoints of the edges containing
$v$ changes which requires the priority queue to be updated
accordingly. Similarly, if a birth or death event occurs the edges of
$\contourtree$ change locally around $v$ and $u$, and the priority
queue needs to be updated accordingly.

%% file: update_mesh.tex
\subsection{High-Level Operations}\label{ref:update_mesh}

We describe high-level operations to update the triangulation
$\mesh$. In particular we show how to insert and delete vertices and
how to perform edge flips. These operations are sufficient for
supporting a wide number of important algorithms including
e.g. algorithms for maintaining constrained delaynay triangulations.

\begin{figure}
\centering{
\begin{tabular}{ccc}
\includegraphics[height=1.2in,page=1]{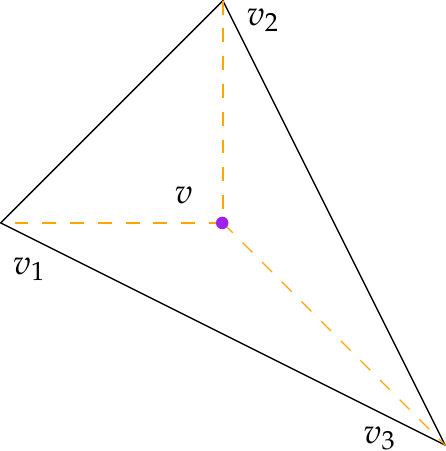}&
\includegraphics[height=1.2in,page=3]{figures/update_mesh}\\
(a)&(c)\\
\end{tabular}
}
\caption{Updating $\mesh$. (a) Inserting/deleting vertex $v$, (b)
  inserting $v$ on an edge $e$ and flipping edge
  $e=v_1v_2$.}
\label{fig:updatemesh}
\end{figure} 

\paragraph{Insert vertex} 
The $\insertv(v,r)$ operation inserts vertex $v\in\R^2-\vertices$ with
height $r\in\R$ into $\mesh$. If $v$ is not on an edge of $\mesh$ it
is contained in a triangle $t\in\faces$ with vertices $v_1$, $v_2$ and
$v_3$.  We first insert $v$ in $\mesh$ by connecting it to $v_1$,
$v_2$ and $v_3$, creating three new triangles and removing
$t$. Initially we set the height of $v$ to the elevation of $\terrain$
at $v$ (given by linear interpolation on the vertices of $t$). Note
that no point on $\terrain$ changes height as a result of this update,
and thus $\contourtree$ does not change and neither does any existing
link pointers. To update the descent tree forest of $\terrain$, we
select a vertex $u$ in $\lowerlink(v)$. Assume that $u$ belongs to
$\descendtree(x)$. We then create a descent tree vertex representing
$v$ and make it the child of $u$ in $\descendtree(x)$. Updating the
ascent tree forest is similar. Link pointers from $v$ are created by
looking at vertices in $\link(v)$ consisting of the three vertices of
$t$. Finally we invoke $\changeheight(v,r)$ to set the elevation of
$v$ correctly. Note that this works on the unbounded face as well, by
fixing the elevation of $v_{\infty}$ appropriately. A similar
procedure to the above is used when $v$ lies in the interior of an
edge $e=v_1v_2$ of $\mesh$.  Let $v_3$ and $v_4$ be the vertices
opposite $e$ in $\mesh$, refer to Figure~\ref{fig:updatemesh}(b). In
this case we need to add edges $vv_3$ and $vv_4$ to $\mesh$ and update
the data structure in a way that is similar to what is described
above.

\paragraph{Delete vertex}
The $\delete(v)$ operation deletes vertex $v\in\vertices$ from $\mesh$
where $v\neq v_{\infty}$. We assume that $|\link(v)|=3$. Let
$\link(v)=\{v_1,v_2,v_3\}$ and set $r$ to the height of the triangle
$t=v_1v_2v_3$ at $v$, found by linear interpolation. We first invoke
$\changeheight(v,r)$ to ensure that $v$ lies in the plane of $t$. This
implies that $v$ is not a critical vertex of $\terrain$ and that we
can remove $v$ wthout affecting the height of any point on $\terrain$,
and therefore also without affecting $\contourtree$. Then we simply
remove $v$ from $\mesh$ along with edges $vv_1$, $vv_2$ and $vv_3$ and
their associated faces. Assume $v$ belongs to $\descendtree(x)$, the
neighbors of $v$ in $\descendtree(x)$ either belong to $\lowerlink(v)$
or $\upperlink(v)$. We simply update $\descendtree(x)$ by removing $v$
and connecting neighbors in $\upperlink(v)$ to a neighbor in
$\lowerlink(v)$. Since $v$ is not a minima $\lowerlink(v)$ is always
nonempty. The ascent tree containing $v$ is updated in a similar
way. If either one of $v_1$, $v_2$ or $v_3$ has a lower link pointer
to $v$, we simply replace this pointer with the lower link pointer of
$v$. Upper link pointers are updated in a similar way. Similarly to
the $\insertv$ operation, $\delete(v)$ works if $v$ is on the boundary
of $\mesh$, i.e. when $v$ is adjacent to $v_{\infty}$.

\paragraph{Edge Flip}
The $\edgeflip(e)$ operation flips an edge $e=u_1u_2\in\edges$ of
$\mesh$, where $u_1,u_2\neq v_{\infty}$. Let $v_3$ and $v_4$ be the
vertices opposite $e$ in $\mesh$, refer to
Figure~\ref{fig:updatemesh}(b). If the quadrilateral $v_1v_2v_3v_4$ is
not convex, the operation is undefined. The edge flip is accomplished
by first invoking $\insertv(v,r)$ where $v$ is the intersection of $e$
and the line segment $v_3v_4$ and $r$ is found by linear interpolation
at $v$ across $v_3v_4$. We can then delete vertex $v$ using a slightly
modified version of the $\delete(v)$ operation: instead of deleting
edge $v_3v$ and $vv_4$ we merge them to create one big edge
$e'=v_3v_4$. Due to the selection of $r$, this does not affect the
height of any point on $\terrain$. Similarly to the above two
operations, the edge flip can be made to work if $v_3$ or $v_4$ is
$v_{\infty}$.